\documentclass[10pt,journal,twoside]{IEEEtran}

\usepackage[numbers,sort&compress]{natbib}
\usepackage{graphicx}
\usepackage{subfigure}
\usepackage{verbatim}
\usepackage{multirow}
\usepackage{amsmath}
\usepackage{amsfonts}
\usepackage{amssymb}
\usepackage{dsfont}
\usepackage{setspace}
\usepackage[ruled,vlined,linesnumbered]{algorithm2e}

\usepackage{balance}
\usepackage{amsthm}
\newtheorem{theorem}{Theorem}
\newtheorem{remark}{Remark}

\newtheorem{definition}{Definition}

\usepackage{multicol}
\usepackage{amsfonts}
\usepackage{booktabs}
\usepackage{multirow}
\usepackage{threeparttable}
\usepackage[usenames,dvipsnames]{color}
\usepackage{colortbl}

\newcommand{\ignore}[1]{{}}

\makeatletter

\newcommand{\Rmnum}[1]{\expandafter\@slowromancap\romannumeral #1@}
\makeatother

\usepackage{siunitx}

\usepackage{changes}
\usepackage{cancel}
\setdeletedmarkup{\cancel{#1}}

\begin{document}
\title{\huge Online Digital Twin-Empowered Content Resale Mechanism in Age of Information-Aware Edge Caching Networks}

\author{\IEEEauthorblockN
{Yuhan Yi,~Guanglin Zhang,~and Hai Jiang}
\vspace{-10mm}
\thanks{
Y.~Yi and G.~Zhang are with the College of Information Science and Technology, Donghua University, Shanghai, China (e-mails: yuhanyi@mail.dhu.edu.cn, glzhang@dhu.edu.cn). H.~Jiang is with the Department of Electrical and Computer Engineering,
University of Alberta, Edmonton, Canada (e-mail:
hai1@ualberta.ca). 
}
}

\maketitle

\begin{abstract}
For users requesting popular contents from content providers, edge caching can alleviate backhaul pressure and enhance the quality of experience of users.
Recently there is also a growing concern about content freshness that is quantified by age of information (AoI).
Therefore, AoI-aware online caching algorithms are required, which is challenging because the content popularity is usually unknown in advance and may vary over time.
In this paper, we propose an online digital twin (DT) empowered content resale mechanism in AoI-aware edge caching networks.
We aim to design an optimal two-timescale caching strategy to maximize the utility of an edge network service provider (ENSP). The formulated optimization problem is non-convex and NP-hard.
To tackle this intractable problem, we propose a DT-assisted Online Caching Algorithm (DT-OCA). In specific, we first decompose our formulated problem  into a series of subproblems, each handling a cache period. For each cache period, we use a DT-based prediction method to effectively capture future content popularity, and develop online caching strategy.
Competitive ratio analysis and extensive experimental results demonstrate that our algorithm has promising performance, and outperforms other benchmark algorithms. Insightful observations are also found and discussed.
\end{abstract}
\begin{IEEEkeywords}
Edge caching, age of information, digital twin, competitive ratio, Transformer.
\end{IEEEkeywords}

\IEEEpeerreviewmaketitle

\section{Introduction}\label{Introduction}
In recent years, data traffic has grown at a tremendous rate, in which content requests from users and content delivery by content providers (CPs) account for a significant portion of data \cite{8960482}. The advancement of 5G and the Internet of things (IoT) will accelerate this growth, which puts enormous pressure on the core network and may lead to poor user experience.
To handle this issue, one efficient technique is edge caching, i.e., popular contents are cached in edge nodes such as edge servers (ESs) or base stations \cite{9130754}, and then a large number of content requests can be processed locally at those edge nodes without repeated requests to core networks.

For an edge caching network, a caching strategy is needed, to determine which contents to cache. In the literature, the majority of the caching strategies were designed to maximize the {\it hit rate} of the cached contents.
The work in \cite{9442312} maximized cache hit rate by proposing a two-step framework, in which a suitable coding scheme was adopted first, and then a joint cache placement and recommendation decision algorithm based on the decision-making game was developed.
The work in \cite{9832658} modelled the hit-rate maximization problem as a separable assignment problem and solved it by a recursive enumeration method.
The work in \cite{10123387} studied the edge caching strategies considering the content delivery and cache replacement to improve hit rate, in which a distributed multi-agent deep reinforcement learning (DRL) is utilized.
Recently, content popularity prediction has been adopted in many caching strategy designs, as information of content popularity can help to increase the cache hit rate. The work in \cite{9349200} developed an online regression based on the Gaussian process algorithm to make short-term prediction for contents in vehicular edge caching networks. The work in \cite{10025827} utilized a temporal convolution network and an attention mechanism to learn content popularity patterns for achieving the hit rate maximization. The work in \cite{10129198} utilized a self-attentive sequential recommendation model for content popularity prediction and proposed a collaborative caching scheme. The work in \cite {9896205} used a recommendation system-based prediction model to predict content popularity and then designed a DRL-based reactive caching strategy to increase hit rate and reduce average downloading time.

In addition to traditional requirements for quality of experience (QoE), users today also care about \textit{content freshness}.
Users want fresh breaking news or real-time information, such as traffic conditions, real-time map \cite{8555643}, weather forecast, sales promotions \cite{9777850}, epidemic status, etc., and prefer not to miss out. Over time, however, such information can become outdated and worthless. Therefore, a new metric, namely age of information (AoI), has been proposed to quantify the content freshness. 
For a content, its AoI  describes the time elapsed since its generation \cite{9681851}. A smaller AoI means that the content is fresher. Note that AoI is different from the conventional delay metric that measures the time gap between user raising the request and getting the content \cite{9622883}.

\textcolor{black}{Recently, an increasing number of edge network service providers (ENSPs) are now paying CPs to cache their fresh and popular contents \cite{9415705}. An example CP is the famous TV and movie portal Netflix, which provides fresh and popular contents, including TV series, movies, variety shows, and anime.
Other example cached contents may include real-time map and traffic condition in intelligent transportation systems \cite{8555643}, sensor data in IoT networks \cite{9681851}, and social media, where users always prefer fresh contents. In general, when an ENSP caches contents from CPs, it must take into account both content freshness (evaluated by AoI) and content popularity, to make as many profits as possible. Thus, at an ENSP, an AoI-aware online content caching strategy is required, which answers two major questions: `which contents should be purchased and cached?' and `how long purchased contents should be cached?'. The caching strategy design is challenging, as the first question leads to a spatial coupling due to the limited cache capacity, and the second question and the AoI-awareness lead to a temporal coupling.}

There are limited works on edge caching networks with consideration of content freshness.
The work in \cite{9917351} minimized a joint cost including AoI cost, recommendation cost, and downloading cost by designing a caching/updating strategy based on Lagrangian decomposition.
The work in \cite{9199125} proposed two suboptimal scheduling algorithms via enforced decomposition technique and DRL to minimize the average AoI of the dynamic contents.
The work in \cite{10195938} developed online cache update strategies by using statistical probability and decision reward. The objective was to minimize the application AoI and sensor energy consumption.
The work in \cite{9583884} proposed an algorithm based on subproblem decomposition and convex approximation for providing fresh contents to users in a cache-enabled unmanned aerial vehicle network.
The work in \cite{9415705} used game theoretical approaches to model the interactions between CPs and the ENSP and between the ENSP and users. In the games, the CPs determine content pricing, the ENSP decides which CPs to cache contents from and also determines the prices to resell the cached contents to users based on AoI, and users decide whether or not to purchase contents form the ENSP. In all these works \cite{9917351,9199125,10195938,9583884,9415705}, the popularity of a content is known in advance.

\textcolor{black}{In practical edge caching networks, the popularity of contents may not be known in advance by the system \cite{10012407}, and may also vary over time \cite{9384286}. Further, the CPs may keep generating new contents that may be of interest to users. To address these challenges, an ENSP needs a prediction method to predict content popularity over time, and also needs to monitor the content status at CPs. We propose to use digital twin (DT) to keep monitoring the system status and to predict popularity of contents over time. In the literature, DT has been used extensively (e.g., by Huawei and Nokia \cite{9899718}) to monitor, mimic, and predict physical worlds and facilitate optimal online strategies for physical worlds \cite{9547763, 9447819, 9832511, 9832009, 9399641}. In specific, DT replicates physical world to virtual world and maintains twin mappings consistent with the physical world by monitoring and interacting continually with the physical world. It provides a better way to simulation, analysis, prediction, and decision making, enabling comprehensive insights and more efficient optimizations of the physical world. In addition, DT networks can be further benefited by using the most advanced artificial intelligence (AI) techniques. AI techniques can help DT to achieve  smart data collection and accurate prediction, while DT also provides AI with comprehensive and precise data required for its learning process \cite{9451579}. Thus, we propose to use an AI-aided DT for the ENSP to capture request distribution, model content popularity prediction, and develop an effective online caching strategy. This research topic is still an open issue.}

\textcolor{black}{As for predicting content popularity, plenty of time-series prediction techniques have been proposed, such as statistical methods (e.g., auto-regressive [AR], auto-regressive integrated moving average [ARIMA]), AI methods (e.g., support vector machine [SVM]), decision tree (e.g., XGBoost, LightGBM), deep learning (e.g., recurrent neural networks [RNN], long short-term memory [LSTM], Transformer), etc. As indicated in \cite{SHEN2020302}, real user data are normally non-linear and non-stationary, and thus the above methods except deep learning are unable to extract sufficient data features to get accurate prediction.
Besides, SVM, XGBoost, and LightGBM require more complex manual feature engineering that relies on expert experience. On the other hand, deep learning approaches are able to explore complex data features, and do not require significant manual feature engineering and model development \cite{lim2021time}. Among deep learning approaches, RNN-based models may suffer in modeling long-term dependencies \cite{liu2022non}; the LSTM alleviates this issue but the issue still remains unresolved \cite{zhou2021informer}, and the performance of LSTM may be degraded due to its inability to be parallelized. Transformer, another deep learning approach, has achieved progressive breakthrough especially in time-series prediction \cite{liu2022non,zhou2021informer,Ding2021,wu2020deep}. Benefiting from the stacked structure and attention mechanism, Transformer is perfectly suited for sequence prediction tasks by naturally capturing time dependencies from deep multi-level features \cite{liu2022non}. Many literatures have illustrated its superiority \cite{Ding2021,wu2020deep}. Thus, we propose to use Transformer in our AI-aided DT.}

As a summary, we consider a content resale mechanism in AoI-aware edge caching networks and propose to use Transformer-aided DT to design an online caching strategy. In this mechanism, the ENSP purchases and caches popular and fresh contents from CPs, and resells to users for maximizing its utility.
The primary contributions of this work are as follows:
\begin{itemize}
\item \textbf{A novel model:} This paper is the first to use AI-aided DT to assist AoI-aware edge caching strategy design.
\item \textcolor{black}{\textbf{Problem transformation.}} Our optimization problem is non-convex and NP-hard due to the discrete 0-1 decision variables and the coupling between cache periods.
To tackle this intractable problem, we propose a DT-assisted Online Caching Algorithm (DT-OCA) to solve it, by decomposing the original problem into a series of subproblems, each for a cache period.
\item \textcolor{black}{\textbf{Online near-optimal solution.}} Our DT-OCA finds solutions for individual cache periods by using Transformer-aided DT to predict unknown and varying content popularity. Simulation results show that DT-OCA achieves very similar performance to the case assuming that DT-OCA, when handling a cache period, knows perfect content popularity in advance. This indicates that DT-OCA achieves near-optimal solution for each cache period. 
\item \textbf{Performance guarantee.} For the proposed DT-OCA, we analyze its competitive ratio, which measures the worst-case performance gap of an online algorithm and the offline optimal algorithm. Our analysis shows that the DT-OCA has a very promising competitive ratio.
\item \textbf{Effectiveness and superiority.} We conduct extensive simulations compared with other benchmark algorithms to demonstrate the effectiveness and superiority of the proposed DT-OCA algorithm.
Besides, the impacts of system parameters are analyzed, indicating that DT-OCA performs well in a large range of system parameters.
\item \textbf{Insightful observation.} In our experiments, it is observed that a smaller updating frequency of the DT network generally leads to a lower system utility. However, when the updating frequency is at some specific values, the system utility has a boost, which is counter-intuitive. For this, an insightful discussion is provided.
\end{itemize}

\begin{figure*}[t]
  \centering
  \includegraphics[width=0.75\textwidth]{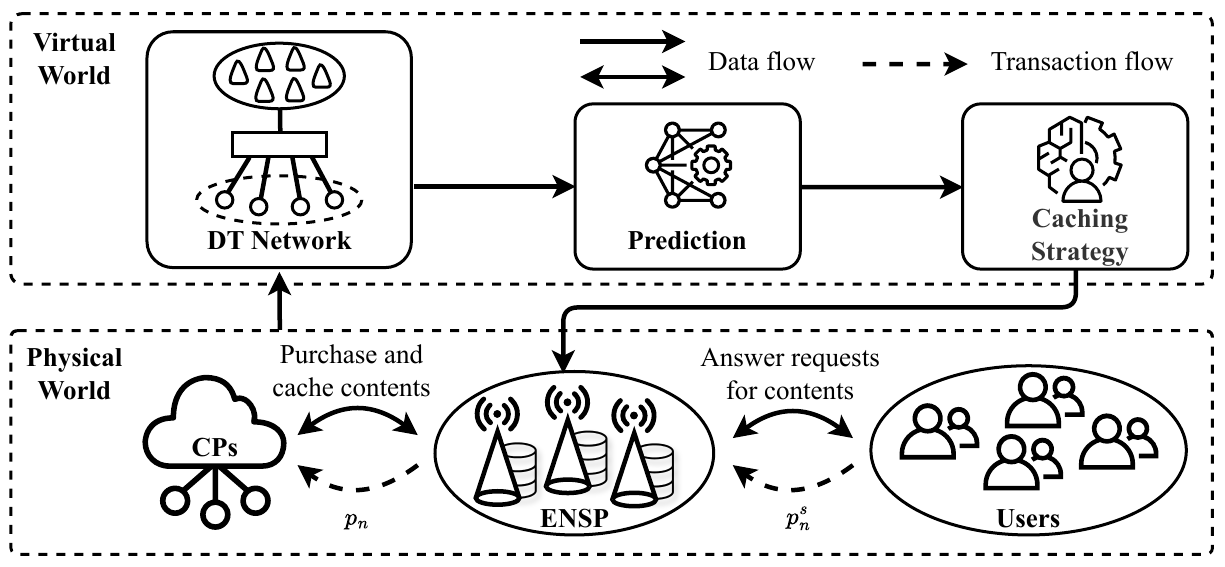}
  \caption{Online DT-empowered content resale mechanism in AoI-aware edge caching networks.}
  \label{Fig:system-overview}
\end{figure*}

 The rest of the paper is arranged as follows.
 Section~\ref{Sec:system model} presents the detailed edge caching system model and formulates a utility-maximization problem, followed by the DT-OCA with analysis in Section~\ref{Sec:algorithm}. Section~\ref{Sec:simulation} shows the simulation results of DT-OCA and other benchmark algorithms. Section~\ref{Sec:conclusion} concludes the paper.

\section{System Model and Problem Formulation}\label{Sec:system model}
We consider an online DT-empowered content resale mechanism in AoI-aware edge caching networks including the physical world and virtual world, shown in Fig.~\ref{Fig:system-overview}.
By mapping the physical world to the virtual world through DT, the environment in the realistic caching network could be reproduced and simulated.
Specifically, the physical world consists of several CPs, an ENSP, and a number of users, which can be seen as a commercial market with three types of participants, as follows.
\begin{itemize}
\item \textit{Content Provider}: CPs play the role of sellers in the market. They provide the ENSP with a large amount of randomly generated contents (such as news, sensor data, etc.), and charge corresponding fees. The price of a content is related to its generation cost.
\item \textit{Edge Network Service Provider}: The ENSP is a middleman in the market and is responsible for a group of ESs with limited cache capacities. After contents are purchased from CPs and cached in ESs, the ENSP resells them to users, which reduces the backhaul cost by avoiding repeatedly fetching the same content from CPs.
\item \textit{User}: Users act as consumers in the market, which initiate content requests to the ENSP and pay the corresponding service fees. They are interested in fresh contents, and thus, are willing to pay a higher service fee for the content with a smaller AoI.
\end{itemize}

The virtual world consists of three modules, i.e., \textit{DT network}, \textit{prediction}, and \textit{caching strategy} as follows.
\begin{itemize}
\item \textit{DT Network}: The DT network is built on the ESs and provides a virtual representation of the physical edge caching network. It obtains the precise status of the physical world and tracks its variations.
\item \textit{Prediction}: The prediction module is constructed to extract key features of diverse content requests and predict future content popularity. It could capture the user preferences of contents in the dynamic edge caching network and enhance caching strategy performance.
\item \textit{Caching strategy}: The caching strategy is to maximize the utility of the ENSP. It takes into account new content generation\footnote{\textcolor{black}{In the considered system, CPs keep generating new contents. The new contents can be viewed as updated versions of old contents.}}, limited cache capacities, future content popularity, and user requirements on content freshness. With the help of the DT network, the caching strategy can be optimized in the virtual world and sent to the physical world for implementation.
\end{itemize}

In this paper, we consider two timescales: time slot and cache period. A time slot has a duration $\tau$. A cache period includes $b$ time slots ($b$ being a positive integer), and thus, has a duration $b\tau$. Let $\mathcal{T}=\{0,1,...,T\}$ denote the set of time slot indices, and $\mathcal{L}=\{0,1,...,L\}$ denote the set of cache period indices. The mutual conversion between time slot $t$ and cache period $l$ satisfies $l=\lfloor t/b \rfloor$ and $t=b l+d$, in which $d$ is an integer and $0 \leq d < b$. In particular, the first and last time slot in cache period $l$ have time slot indices $bl$ and $bl+(b-1)$, respectively.

Denote $\mathcal{N}$ as the set of all contents generated by CPs in the system. The content $n \in \mathcal{N}$ is described as $content_n=\{t_n, s_n, p_n, A_n(t)\}$, where $t_n$ means that content $n$ is generated at time slot $t_n$, $s_n$ represents the size of content $n$, $p_n$ is the price of content $n$ sold by CPs to the ENSP, and $A_n(t)$ is the AoI of content $n$ at time slot $t$:
\begin{equation}
A_n(t)=
\begin{cases}
t-t_n,&\mbox{$t>t_n$}\\
0,&\mbox{$t\leq t_n$.}
\end{cases}
\label{eq: AoI}
\end{equation}

\subsection{Two-Timescale Caching Strategy}\label{sec:cache_strategy}
The \textit{caching strategy} including purchasing decision and caching decision is based on two timescales, as follows.
\begin{itemize}
\item The ENSP's purchasing decision is based on the timescale of a cache period. In specific, $x_n(l)\in\{0,1\}$ denotes the purchasing decision of the ENSP for content $n$ in cache period $l$: $x_n(l)=0$ means that the ENSP does not purchase the content at the beginning of cache period $l$, while $x_n(l)=1$ means that the ENSP purchases the content at the beginning of cache period $l$.
\item The ENSP's caching decision is based on the timescale of a time slot. In specific, $y_n(t)\in\{0,1\}$ denotes the caching decision of the ENSP for content $n$ at time slot $t$: $y_n(t)=0$ means that the ENSP does not cache the content at the beginning of time slot $t$, while $y_n(t)=1$ means that the ENSP caches the content at the beginning of time slot $t$.\footnote{\textcolor{black}{The reason of using two timescales is as follows. As to be shown in Section \ref{Sec:algorithm}, the ENSP will use the DT-based prediction method to keep predicting the popularity of contents. It is possible that the predicted popularity of a cached content may become lower in the middle of a cache period. For such cases, it may be better to release the content in the middle of the cache period, to save caching cost. Thus, for the caching decision, it may not be appropriate to use the timescale of a cache period. So the caching decision is based on a smaller timescale, i.e., the timescale of a time slot. Nevertheless, if the ENSP sets $b=1$, the two-timescale system reduces to a one-timescale system.}}
\end{itemize}
Contents can only be purchased at the beginning of each cache period $l$, and the purchased contents would be cached immediately until released. If a released content needs to be re-cached, it needs to be purchased again.
Caching decision $y_n(t)$ may vary across time slots in a cache period, as follows.
\begin{itemize}
\item If $x_n(l)=1$ (i.e., the ENSP decides to purchase content $n$ at the beginning of cache period $l$), we have $y_n(bl)=1$, i.e., the content should be cached at the first time slot of cache period $l$. For subsequent time slots in cache period $l$, the ENSP decides whether to keep caching the content (i.e., $y_n(t)=1$) or release the content (i.e., $y_n(t)=0$). If content $n$ is released at a time slot, we always have $y_n(t)=0$ for remaining time slots in cache period $l$.

\item If $x_n(l)=0$ and $y_n(bl-1)=1$ (i.e., content $n$ is cached at the last time slot in the previous cache period), the ENSP decides whether to keep caching the content (i.e., $y_n(t)=1$) or release the content (i.e., $y_n(t)=0$) for each time slot in cache period $l$. If content $n$ is released at a time slot, we always have $y_n(t)=0$ for remaining time slots in cache period $l$.

\item If $x_n(l)=0$ and $y_n(bl-1)=0$ (i.e., content $n$ is not cached at the last time slot in the previous cache period), we always have $y_n(t)=0$ for all time slots in cache period $l$.

\end{itemize}


We denote the {\it caching duration} of content $n$ in cache period $l$ as $v_n(l)=\sum_{t\in \boldsymbol{l}} y_n(t)$, where $\boldsymbol{l}\triangleq \{bl, bl+1,..., bl+(b-1)\}$ means the set of time slots in cache period $l$. We have $0\leq v_n(l) \leq b,~\forall n, l$.

Define \textit{purchasable content} as a content that has been generated by CPs and has AoI no more than a threshold $\phi$.
Define $\mathcal{N}_p=\{n \mid 0< A_{n}(bl) \leq \phi, n\in \mathcal{N}\}$ as the set of purchasable contents for cache period $l$, and $\mathcal{N}_c=\{n \mid y_{n}(t)=1, n\in \mathcal{N}\}$ as the set of contents currently cached by the ENSP at time slot $t$. \textcolor{black}{For presentation simplicity, we omit the label $l$ for  notation $\mathcal{N}_p$ and label $t$ for notation $\mathcal{N}_c$.}
Thus, we have
\begin{align}
x_{n}(l)=0,~\forall l, n\notin \mathcal{N}_p.
\label{eq: x constraint}
\end{align}

Denote $S(t)\triangleq\sum_{n\in \mathcal{N}} y_n(t) s_n$ as the total size of contents cached by the ENSP at time slot $t$, which is limited by the maximum cache size $S_{\max}$ of the ENSP, i.e.,
\begin{align}
S(t) \leq S_{\max},~\forall t.
\label{eq: size constraint}
\end{align}



\subsection{Utility Model}
Users care about the freshness of contents. As a result, the service fees paid by users to the ENSP are based on how fresh the contents can remain, which can be measured by average AoI.
If a request\footnote{We assume that if the request is raised before the ENSP purchases the content or after the ENSP releases its cache, the ENSP responds to the request by fetching the content from the CPs, which is out of the scope of this paper.} for content $n$ arrives at the $m$th time slot, the cached content will be delivered to the user at the next time slot with AoI equal to $A_n(m+1)$.

\textcolor{black}{At the beginning of each cache period, the ENSP will announce a service fee for content $n$ and the service fee is fixed for the whole cache period. Intuitively, for cache period $l$, the ENSP sets the service fee as a linear function of the average AoI that all user requests for content $n$ can achieve in cache period $l$.\footnote{\textcolor{black}{Our proposed algorithm can still work if a different setting of the service fee is used, because the form of the service fee does not affect the algorithm design.}} For cache period $l$, the average AoI of content $n$, denoted by $\overline{A_n}(l)$, is expressed as
\begin{align}
\overline{A_n}(l)=\frac{1}{R_n(l)}\sum_{t\in \boldsymbol{l}} r_n(t)A_n(t+1),
\label{average AoI}
\end{align}
where $r_n(t)$ is the {\it content popularity} of content $n$ at time slot $t$ (defined as the number of user requests for content $n$ at time slot $t$), expression of $A_n(\cdot)$ is given in (\ref{eq: AoI}), and $R_n(l)\triangleq\sum_{t\in \boldsymbol{l}} r_n(t)$ is the total number of user requests for content $n$ during cache period $l$.
The service fee for content $n$ in cache period $l$, denoted by $p_n^s(l)$, is set up as
\begin{align}
p_n^s(l)=q_n^{\max}-\lambda \overline{A_n}(l),
\label{service fee unknown}
\end{align}
where $q_n^{\max}$ is the maximum service fee of content $n$ and $\lambda$ is the weight parameter.}

\textcolor{black}{However, the service fee should be announced at the beginning of cache period $l$, at which moment the ENSP does not know $r_n(t)$ in time slots of cache period $l$. Since  the ENSP could calculate $\overline{A_n}(l-1)$ (the average AoI in the previous cache period) at the beginning of cache period $l$, we approximately use $\overline{A_n}(l) \approx \overline{A_n}(l-1) + b$ (unit: time slot).
Accordingly, service fee $p_n^s(l)$ is expressed as
\begin{align}
p_n^s(l)=p_n^{\max} -\lambda \overline{A_n}(l-1),
\label{service fee known}
\end{align}
in which $p_n^{\max} \triangleq q_n^{\max}-\lambda b$.
}


Denote $u_n(l)$ as the utility of the ENSP for content $n\in\mathcal{N}_p\cup\mathcal{N}_c$ in cache period $l$, and we have
\begin{equation}
\begin{split}
u_n(l)=&\sum_{t\in \boldsymbol{l}}\big\{y_n(t)r_n(t)[p_n^s(l)+s_n C_d]-y_n(t)s_n C_a\big\}\\
&-x_n(l)(p_n+s_n C_d),
\label{eq: u_n(l)}
\end{split}
\end{equation}
where $C_d$ is the transmission cost per unit size from CPs to the ENSP, and $C_a$ is the caching cost per unit size per time slot. In the expression, the term $y_n(t)r_n(t)[p_n^s(l)+s_n C_d]$ means the profit earned by the ENSP at time slot $t$, which includes the service fees paid by users and the backhaul transmission savings due to edge caching; the term $y_n(t)s_n C_a$ is the cost of caching content at time slot $t$; and the term $x_n(l)(p_n+s_n C_d)$ is the cost of purchasing content $n$ and downloading it from the CPs in cache period $l$.

\subsection{Problem Formulation}
The objective of the system is to maximize the total utility of the ENSP, which is formulated as problem $\mathcal{P}1$ as follows:
\begin{align}
\mathcal{P}1:~
&\max_{\mathbf{x},\mathbf{y}}~\sum_{l\in \mathcal{L}}\sum_{n\in \mathcal{N}} u_n(l)\nonumber\\
\text{s.t.}~~&(\ref{eq: x constraint}),(\ref{eq: size constraint}),\notag \\
&x_n(l), y_n(t) \in \{0,1\},~\forall n, l, t,\label{eq: problem c1}
\end{align}
where $\mathbf{x}\triangleq [x_n(l)]_{n\in\mathcal{N}, l\in\mathcal{L}}$ and $\mathbf{y}\triangleq [y_n(t)]_{n\in\mathcal{N},t\in\mathcal{T}}$.

For the ENSP, the caching strategy includes purchasing decisions $x_n(l)$ and caching decisions $y_n(t)$, which are interactive and coupled. The purchasing decisions determine which contents could be cached later. On the other hand, the caching decisions affect the future purchasing decisions of other contents due to the cache capacity constraint (\ref{eq: size constraint}).

\begin{remark}
Problem $\mathcal{P}1$ is non-convex, as the discrete 0-1 decision variables ($x_n(l)$ and $y_n(t)$) lead to a non-convex set of the domain. Moreover, it is an NP-hard problem.
\end{remark}

\textcolor{black}{To solve the NP-hard problem $\mathcal{P}1$, next we will transfer it to subproblems, each considering only one cache period. Then, each subproblem will be transformed to a 0-1 knapsack problem, which will be described in details in Section \ref{sec:problem_reformulation}}.

\begin{figure*}[t]
  \centering
  \includegraphics[width=0.75\textwidth]{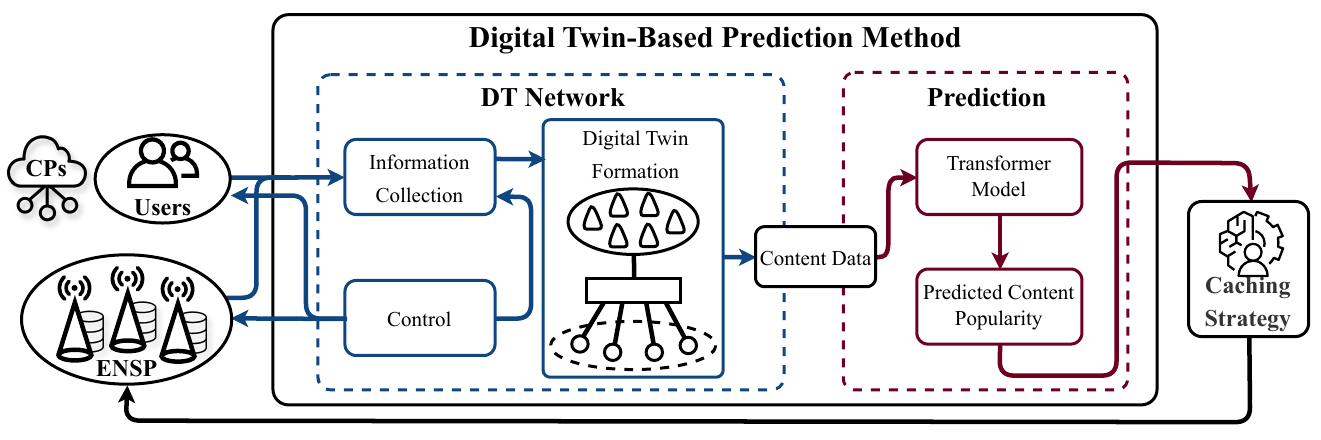}
  \caption{DT-based prediction method.}
  \label{Fig:DT-pipeline}
\end{figure*}

\section{Digital Twin-Assisted Online Caching Algorithm}\label{Sec:algorithm}
This section presents our proposed DT-OCA. We first decompose problem $\mathcal{P}1$ into a series of subproblems $\mathcal{P}2$, each for a cache period. We also propose an offline algorithm to solve problem $\mathcal{P}2$. The offline algorithm needs information of future content popularity, which is unknown in a practical system. Thus, we use a DT-based prediction method to predict content popularity and develop the online caching strategy.
\textcolor{black}{We derive the competitive ratio between the proposed DT-OCA and the offline optimal algorithm for problem $\mathcal{P}1$, to evaluate the worst-case performance of the proposed algorithm.}
\subsection{Problem Decomposition and Offline Algorithm}\label{sec:problem_reformulation}
To develop an efficient solution, we first decompose problem $\mathcal{P}1$ into $L$ subproblems (called problem $\mathcal{P}2$), each for one cache period. The subproblem for cache period $l$ is as follows:
\begin{align}
\mathcal{P}2:~
&\max_{\mathbf{x}(l),\textcolor{black}{\mathbf{y}(t),t\in\boldsymbol{l}}}~\sum_{n\in \mathcal{N}} u_n(l)\nonumber\\
\text{s.t.}~~&(\ref{eq: x constraint}),(\ref{eq: size constraint}),(\ref{eq: problem c1})\notag
\end{align}
in which $\mathbf{x}(l)\triangleq [x_n(l)]_{n\in\mathcal{N}}$ and $\mathbf{y}(t)\triangleq [y_n(t)]_{n\in\mathcal{N}}$.


Then, we develop an offline algorithm to solve the above problem $\mathcal{P}2$.

For content $n$, its caching decision for the first slot of cache period $l$ (i.e., time slot $bl$) is denoted by $y_n(bl)$. For the first time slot, the cache capacity constraint (\ref{eq: size constraint}) is expressed as
\begin{equation}\label{e:first_slot_cap}
    \sum_{n\in \mathcal{N}} y_n(bl)s_n \le S_{\max}.
\end{equation}

For content $n$, if $y_n(bl)=0$ (i.e., content $n$ is not cached at the first time slot of cache period $l$), then we have $u_n(l)=0$ (i.e., content $n$ does not have utility in cache period $l$).

For content $n$, if $y_n(bl)=1$ (i.e., content $n$ is cached at the first time slot of cache period $l$), we have two cases as follows.
\begin{itemize}
    \item  Case I: Content $n$ is not cached at the end of the previous cache period (i.e., cache period $l-1$), which means that content $n$ is purchased and cached at the beginning of cache period $l$. So for content $n$, the caching strategy for the first time slot is: $x_n(l)=1, y_n(bl)=1$.

Recall that in the first time slot, constraint (\ref{e:first_slot_cap}) needs to be satisfied. As long as constraint (\ref{e:first_slot_cap})  is satisfied at the first time slot, then the cache capacity constraint (\ref{eq: size constraint}) is always (automatically) satisfied in all subsequent time slots of cache period $l$, since the ENSP will not cache new contents in subsequent time slots. Thus, for the system, when the caching strategies of all contents for the first time slot are given, the caching decisions of all cached contents in subsequent time slots are not coupled any more since the cache capacity constraint (\ref{eq: size constraint}) is automatically satisfied. Thus, in order to maximize the system utility, the cached contents' caching decisions for subsequent time slots can be derived separately such that each content's individual utility in cache period $l$ is maximized. In other words, in the optimal solution of problem $\mathcal{P}2$, content $n$'s utility in cache period $l$, denoted by $u_{n}^{\max}(l)$, is expressed as
\begin{equation}\label{eq: maximum potential utility_caseI}
u_{n}^{\max}(l)=\max_{\substack{{x_n(l)=1}, y_n(bl)=1,\\ y_n(t), t\in\boldsymbol{l}\backslash \{bl\}}} u_n(l),~n\in\mathcal{N}_p\backslash\mathcal{N}_c,
\end{equation}
and content $n$'s caching strategies at all time slots (except the first time slot) of cache period $l$ are the solution for (\ref{eq: maximum potential utility_caseI}).

\item Case II: Content $n$ is cached at the end of the previous cache period (i.e., cache period $l-1$), which means that content $n$'s caching strategy for the first time slot in cache period $l$ is: $x_n(l)=0, y_n(bl)=1$. Similar to the discussion in Case I, the cached contents' caching decisions for subsequent time slots can be derived separately such that each content's individual utility in cache period $l$ is maximized. So in the optimal solution of problem $\mathcal{P}2$, content $n$'s utility is expressed as
\begin{equation}\label{eq: maximum potential utility_caseII}
u_{n}^{\max}(l)=\max_{\substack{{x_n(l)=0}, y_n(bl)=1,\\ y_n(t), t\in\boldsymbol{l}\backslash \{bl\}}} u_n(l),~n\in\mathcal{N}_c,
\end{equation}
and content $n$'s caching strategies at all time slots (except the first time slot) of cache period $l$ are the solution for (\ref{eq: maximum potential utility_caseII}). Note that for presentation simplicity, here we use the same notation $u_{n}^{\max}(l)$ for Case I and Case II.

\end{itemize}

Here we assume that the offline algorithm knows future popularity of content $n$ (i.e., the number of requests for content $n$ at each time slot in cache period $l$). Based on the future content popularity information, $u_{n}^{\max}(l)$ in (\ref{eq: maximum potential utility_caseI}) and (\ref{eq: maximum potential utility_caseII}) can be evaluated.

As a summary of the above discussions, in the optimal solution of problem $\mathcal{P}2$, when $y_n(bl)=0$, then content $n$'s utility in cache period $l$ is zero; when $y_n(bl)=1$, then content $n$'s utility in cache period $l$ is $u_{n}^{\max}(l)$. In other words, the optimal utility of problem $\mathcal{P}2$ is expressed as $\sum_{n\in \mathcal{N}}y_n(bl)u_n^{\max}(l)$. Thus, problem $\mathcal{P}2$ is equivalent to a problem that finds out $y_n(bl)$'s for $n\in\mathcal{N}$ such that $\sum_{n\in \mathcal{N}}y_n(bl)u_n^{\max}(l)$ is maximized. In specific, by denoting $z_n(l) \triangleq y_n(bl)$ and
$\mathbf{z}(l)\triangleq [z_n(l)]_{n\in\mathcal{N}}$, problem $\mathcal{P}2$ is equivalent to the following problem, named problem $\mathcal{KP}$:
\begin{align}
\mathcal{KP}:~
&\max_{\mathbf{z}(l)}~\sum_{n\in \mathcal{N}}z_n(l)u_n^{\max}(l) \nonumber\\
\text{s.t.}~~&z_n(l)\in \{0,1\},\forall n, l ~\text{and}~
\sum_{n\in \mathcal{N}} z_n(l)s_n\leq S_{\max},\forall l,\nonumber
\end{align}
in which constraint $\sum_{n\in \mathcal{N}} z_n(l)s_n\leq S_{\max}$ is from (\ref{e:first_slot_cap}).

Problem $\mathcal{KP}$ falls into the classical 0-1 knapsack problem \cite{10.5555/1614191} described as follows. A knapsack can carry weight up to $S_{\max}$. Item $n$ has its value $u_{n}^{\max}(l)$ and weight $s_n$. We should select items to be placed in the knapsack to maximize the total value of the selected items while guaranteeing that the total weight of the selected items is not more than $S_{\max}$.

As a 0-1 knapsack problem, Problem $\mathcal{KP}$ can be solved by the dynamic programming method \cite{10.5555/1614191}.

\subsection{Digital Twin-Based Prediction Method}\label{sec:DTPM}
\textcolor{black}{In Section \ref{sec:problem_reformulation}, to evaluate (\ref{eq: maximum potential utility_caseI}) and (\ref{eq: maximum potential utility_caseII}), the developed offline algorithm assumes perfect knowledge of future content popularity. However, in practical scenarios, content popularity is usually unknown and varies over time.}
To address this challenge, we design a DT-based prediction method, as shown in Fig.~\ref{Fig:DT-pipeline}, to predict future content popularity by exploring the correlation between future and past content popularity.
The proposed method contains two components: \textit{DT network} and \textit{prediction}.

\textbf{DT network}: DT network collects data from CPs, the ENSP, and users, to construct the virtual network.
It contains three modules:
(\romannumeral1) The \textit{information collection module} gathers the content information (including generation time, size, price, and AoI) from CPs, historical content popularity that is measured by the number of content requests from users, and cache status from the ENSP.
(\romannumeral2) The \textit{DT formation module} builds and stores the DT model, which is periodically updated with gathered data by the information collection module.
(\romannumeral3) The \textit{control module} decides the updating frequency in the DT formation module, and the data type in the information collection, and sends the decisions to CPs, ENSP, and users.

\textbf{Prediction}:
The prediction model is executed for predicting future content popularity. At the beginning of each cache period, the prediction model is executed, and based on the prediction results, the ENSP makes purchasing and caching decisions. Further, within each cache period, when the DT network is updated, the prediction model is also executed, and based on the prediction results, the ENSP updates caching decisions.

The prediction of content popularity falls into the field of time series analysis. In this work, we resort to a variant of Transformer \cite{transformer} with an encoder-only structure in the DT-based prediction method to more effectively explore the inherent correlation and complicated structure in content popularity time series. The Transformer is designed to efficiently process sequential data in parallel relying on attention mechanisms.
The attention mechanism of Transformer is inspired by the attention mechanism of human vision \cite{pmlr-v37-xuc15} and concentrates on the relevance of each time step in a sequence. It considerably enhances the model's performance, while also facilitating easier interpretation \cite{8099577}.
The Transformer model of the DT-based prediction method contains three layers: \textit{input layer}, \textit{encoder layer}, and \textit{output layer}, as depicted in Fig.~\ref{Fig:TRM} \cite{transformer}.
\begin{itemize}
\item \textit{Input layer}: It consists of an embedding layer (which projects the input to high-dimensional space) and positional encoding (which utilizes the order of sequence).
\item \textit{Encoder layer}: There are stacked $H$ Transformer blocks including multi-head attention mechanism with position-wise fully connected feed-forward networks (FFNs). Around each attention mechanism and FFN, the residual connection is applied to mitigate gradient vanishing, and then layer normalization is performed.
\item \textit{Output layer}: The output layer takes the output of the encoder layer, passes it through a linear network, and outputs the predicted content popularity results.
\end{itemize}
In the Transformer, the attention mechanism is implemented by the multi-head attention in the encoder layer. The multi-head attention mechanism includes a number ($h$) of single-head attention blocks, expecting that the $h$ single-head attention blocks may extract multiple different attention features of interest. The embeddings from the input layer are linearly projected $h$ times to get inputs of the $h$ single-head attention blocks, which are subsequently calculated in parallel by attention function, are concatenated, and get the output of the multi-head attention after a final projection.

\begin{figure}[t]
  \centering
  \includegraphics[width=0.3\textwidth]{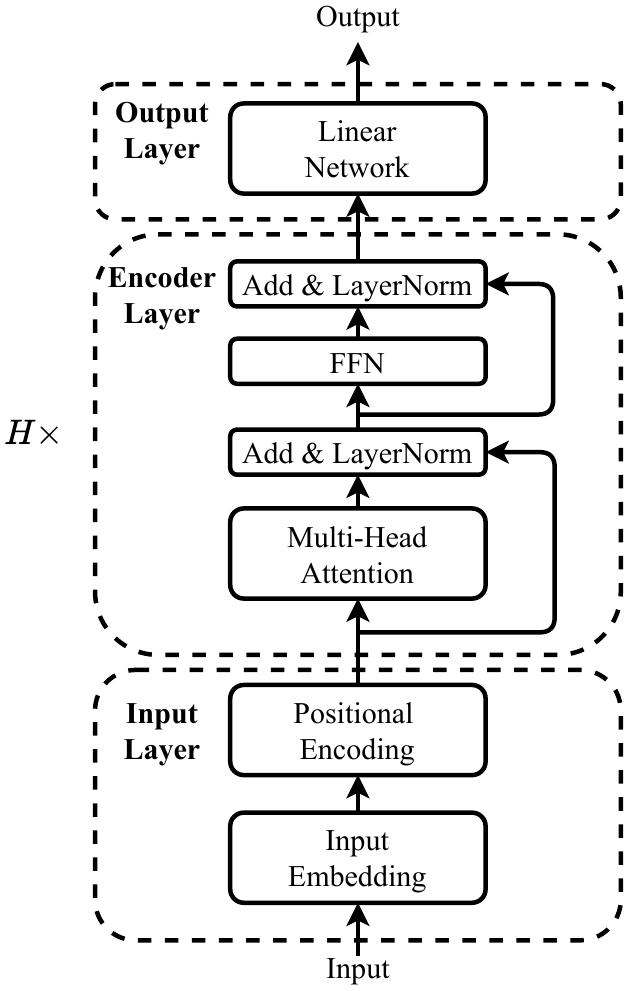}
  \caption{The Transformer model of DT-based prediction method.} \label{Fig:TRM}
\end{figure}


\subsection{Digital Twin-Assisted Online Caching Algorithm}\label{sec:DTOCA}
Based on the content popularity prediction results, next we design an online algorithm, named DT-OCA. DT-OCA will solve problem $\mathcal{P}1$ by giving a caching strategy for problem $\mathcal{P}2$ of each cache period. Next we show how our DT-OCA gives a caching strategy for cache period $l$.

At the beginning of cache period $l$, the ENSP needs to decide which contents to purchase from the purchasable content set. Recall that the DT network is {\it periodically} updated with the physical world. Due to the periodical updating, the virtual world may not be an identical mirror of the physical world at all times. The less frequently the DT network is updated, the less accurate the virtual network will be.
Due to the updating frequency, at the beginning of cache period $l$, the ENSP may not know purchasable content set $\mathcal{N}_p$ at this moment. The ENSP only knows the purchasable content set during the latest updating, denoted by $\mathcal{N}_p'$. Further, some contents inside $\mathcal{N}_p'$ may become stale (i.e., its AoI is more than threshold $\phi$). Thus, the ENSP will remove those contents from $\mathcal{N}_p'$. So at the beginning of cache period $l$, the ENSP's knowledge of purchasable content set in our algorithm is expressed as $\mathcal{N}_p^{\dag}\triangleq\mathcal{N}_p'\backslash \{n\mid A_{n}(bl)>\phi,n\in\mathcal{N}_p'\}$, in which $\{n\mid A_{n}(bl)>\phi,n\in\mathcal{N}_p'\}$ includes the contents that are stale and should be removed. In the sequel, when a notation has a superscript $\dag$, it means a notation for our online algorithm (i.e., DT-OCA).

In cache period $l$, we denote $r_n^{\dag}(t)$ as the predicted number of requests for content $n$ at time slot $t$.
Now we evaluate $u_n^{\dag,\max}(l)$, which is the online-algorithm version of $u_n^{\max}(l)$ given in (\ref{eq: maximum potential utility_caseI}) and (\ref{eq: maximum potential utility_caseII}).
\begin{itemize}
\item For content in the purchasable content set $\mathcal{N}_p^{\dag}$ but not in the cached content set $\mathcal{N}_c$, the utility of content $n$ in cache period $l$ in DT-OCA, denoted by $u_n^{\dag}(l)$, is given as (similar to (\ref{eq: u_n(l)}))
\begin{equation}\nonumber
\begin{split}
u_n^{\dag}(l)=&\sum_{t\in \boldsymbol{l}}\big\{y_n(t)\textcolor{black}{r_n^{\dag}(t)}[p_n^s(l)+s_n C_d]-y_n(t)s_n C_a\big\}\\
&-x_n(l)(p_n+s_n C_d),~n\in\mathcal{N}_p^{\dag}\backslash\mathcal{N}_c.
\end{split}
\end{equation}
\textcolor{black}{Accordingly, the maximum utility of content $n$ in cache period $l$ in DT-OCA, denoted by $u_{n}^{\dag,\max}(l)$, is given as
\begin{equation}\label{eq: ufmax}
u_{n}^{\dag,\max}(l)=\max_{\substack{{x_n(l)=1}, y_n(bl)=1,\\y_n(t), t\in\boldsymbol{l}\backslash \{bl\}}} u_n^{\dag}(l),~n\in\mathcal{N}_p^{\dag}\backslash\mathcal{N}_c.
\end{equation}
}

\item For content in the cached content set $\mathcal{N}_c$, similarly we have
\begin{equation}\nonumber
\begin{split}
u_n^{\dag}(l)=&\sum_{t\in \boldsymbol{l}}\big\{y_n(t)\textcolor{black}{r_n^{\dag}(t)}[p_n^s(l)+s_n C_d]\\
&-y_n(t)s_n C_a\big\},~n\in\mathcal{N}_c.
\end{split}
\end{equation}
\textcolor{black}{Accordingly, the maximum utility of content $n$ in cache period $l$ in DT-OCA is
\begin{align}\label{eq: ufmax_cached}
u_{n}^{\dag,\max}(l)=\max_{\substack{{x_n(l)=0}, y_n(bl)=1,\\ y_n(t), t\in\boldsymbol{l}\backslash \{bl\}}}u_n^{\dag}(l),~n\in\mathcal{N}_c.
\end{align}}
\end{itemize}

At the beginning of cache period $l$, based on the evaluated $u_{n}^{\dag,\max}(l)$, DT-OCA can find caching strategy for cache period $l$ (i.e., purchasing decisions in cache period $l$ and caching decisions for all time slot in cache period $l$) via solving the problem $\mathcal{KP}$ as shown in Section \ref{sec:problem_reformulation}, by replacing $u_{n}^{\max}(l)$ with $u_{n}^{\dag,\max}(l)$. Then the ENSP will perform purchasing and implement the caching decisions for all time slots until the next DT update.

Upon a DT update, the DT-based prediction method will be implemented again, and output updated future content popularity prediction. Based on the prediction results, DT-OCA will update caching decisions for the remaining time slots in cache period $l$, as follows. Denote the set of the remaining time slots in cache period $l$ as $\hat{\boldsymbol{l}}$. The utility of content $n$ for the remaining time slots in cache period $l$, denoted by $\hat{u}_n^{\dag}(l)$, is expressed as
\begin{equation}\label{u_max_remain}
\begin{split}
\hat{u}_n^{\dag}(l)=&\sum_{t\in\hat{\boldsymbol{l}}}\big\{y_n(t)\textcolor{black}{r_n^{\dag}(t)}[p_n^s(l)+s_n C_d]\\
&-y_n(t)s_n C_a\big\},~n\in\mathcal{N}_c.
\end{split}
\end{equation}
The DT-OCA will update caching decisions of content $n$ for the remaining time slots in cache period $l$ such that $\hat{u}_n^{\dag}(l)$ is maximized. This can be done by calculating $\hat{u}_n^{\dag}(l)$ values for all possible content-releasing moments and picking up the maximal $\hat{u}_n^{\dag}(l)$. Then the ENSP will implement the caching decisions for the remaining time slots until the next DT update. The procedure will be repeated until the end of cache period $l$.



The specific procedure of DT-OCA is presented in Algorithm \ref{A1}. Steps 3--8 are for purchasing and caching decisions at the beginning of a cache period, while steps 9--14 are for each time slot to 1) update caching decisions in response to a possible DT network update and 2) implement the most recent caching decisions.

\begin{algorithm}[t]
 \caption{DT-OCA}\label{A1}
 \textbf{Initialization:}
    $\mathcal{N}_p^{\dag}=\varnothing$, $\mathcal{N}_c=\varnothing$\;
 \ForEach{Cache period $l\in \mathcal{L}$}{
     Update $\mathcal{N}_p^{\dag}$ and $\mathcal{N}_c$\;
     \ForEach{Content $n\in \mathcal{N}_p^{\dag}\cup \mathcal{N}_c$}{
        Obtain $r_n^{\dag}(t)$ from DT-based prediction method\;
        Calculate $u_{n}^{\dag,\max}(l)$ according to (\ref{eq: ufmax}) and (\ref{eq: ufmax_cached})\;
    }
    Based on prediction, find optimal caching strategy $\{\mathbf{x}(l),\textcolor{black}{\mathbf{y}(t),t\in\boldsymbol{l}}\}$ to problem $\mathcal{P}2$ by solving problem $\mathcal{KP}$ with $u_{n}^{\max}(l)$ replaced by $u_{n}^{\dag,\max}(l)$\;
    According to the caching strategy, purchase and cache corresponding contents, and update $\mathcal{N}_c$\;
    \For{time slot $bl+1$ \KwTo $bl+(b-1)$}{
        \If{DT network is updated}{
            \ForEach{Content $n\in \mathcal{N}_c$}{
                Update $r_n^{\dag}(t)$ from DT-based prediction method\;
                Update $y_n(t)$ for the remaining time slots such that (\ref{u_max_remain}) is maximized\;
                }
            }
        Implement the most recent caching decisions for the current time slot\;
        }

    }

\end{algorithm}

\subsection{Competitive Ratio Analysis}\label{sec:CR}
The developed online algorithm DT-OCA is to solve problem $\mathcal{P}1$. Competitive ratio is an important measure to evaluate the performance of an online algorithm. Below we give competitive ratio analysis for our DT-OCA.

The competitive ratio describes the worst-case performance of an online algorithm over all instances, defined as follows:
\begin{definition}[Competitive ratio]\label{definition}
For an online algorithm $a$, the competitive ratio is defined as the maximum ratio of the offline optimal utility to the utility achieved by the online algorithm $a$ over all content generations and request distributions, and is expressed as
\begin{align}
\max_{\omega \in \Omega}~\frac{u^*(\omega)}{u^a(\omega)},
\nonumber
\end{align}
where $\omega \in \Omega$ is a feasible instance of the system model, $\Omega$ refers to the set of all feasible instances including all possible content generations and request distributions, and $u^*(\omega)$ and $u^a(\omega)$ denote the utility obtained by the offline optimal solution and online algorithm $a$, respectively, over instance $\omega$.
\end{definition}
The competitive ratio is not less than $1$. If the competitive ratio is closer to $1$, the online algorithm  performs better.

Analyzing the competitive ratio of the proposed online algorithm DT-OCA is challenging, as it is difficult to quantify the impact of the prediction results involved in DT-OCA on the ENSP's utility.
\textcolor{black}{To address the challenge, we consider an approximation algorithm termed DT-OCA with perfect prediction (DT-OCA-PP).} DT-OCA-PP is generally the same as DT-OCA except that DT-OCA-PP has perfect prediction for future content popularity. We consider using DT-OCA-PP to solve problem $\mathcal{P}1$, and evaluate its competitive ratio in the following Theorem \ref{theorem}.\footnote{\textcolor{black}{Since DT-OCA-PP assumes perfect prediction, it can get the optimal solution of problem $\mathcal{P}2$. Thus, the competitive ratio of DT-OCA-PP is a measure for the worst-case performance gap of problems $\mathcal{P}1$ and $\mathcal{P}2$.}} \textcolor{black}{It will be shown in Section \ref{Sec:simulation} that the DT-OCA and DT-OCA-PP have very similar performance, and thus, the competitive ratio of DT-OCA-PP could be considered as a good approximation of the competitive ratio of DT-OCA.}
By defining $R'_n(l)\triangleq\sum_{t\in \boldsymbol{l}}y_n(t)r_n(t)$, expression (\ref{eq: u_n(l)}) can be rewritten as as $u_n(l)=R'_n(l)[p_n^s(l)+s_n C_d]-v_n(l)s_n C_a-x_n(l) (p_n+s_n C_d)$.

\begin{theorem}\label{theorem}
Denote $u^o(\omega)$ and $u^*(\omega)$ as the utility obtained by executing  DT-OCA-PP and the offline optimal algorithm, respectively, over instance $\omega$.
The competitive ratio of DT-OCA-PP, denoted by $\max_{\omega \in \Omega}~\frac{u^*(\omega)}{u^o(\omega)}$, is bounded as
\begin{equation}\label{eq:theorem}
\max_{\omega \in \Omega}~\frac{u^*(\omega)}{u^o(\omega)} \leq 1+\frac{1}{\alpha R-\beta-1},
\end{equation}
in which $\alpha\triangleq\min\limits_{n\in\mathcal{N}_c, \textcolor{black}{t\in\mathcal{T}}}\frac{p_n^s(l)+s_n C_d}{p_n+s_n C_d}$,
$R\triangleq\min\limits_{n\in\mathcal{N}_c, \textcolor{black}{t\in\mathcal{T}}} R'_n(l)$,
and $\beta\triangleq\max\limits_{n\in \mathcal{N}_c, \textcolor{black}{t\in\mathcal{T}}}\frac{b s_n C_a}{p_n+s_n C_d}$.
\end{theorem}

\begin{proof}
In the proof, for presentation simplicity, we omit ``$(\omega)$" from notations $u^o(\omega)$ and $u^*(\omega)$. The proof consists of two steps. In the first step, we will consider an offline strategy by modifying the offline optimal strategy, referred to as {\it modified offline strategy}, and denote the utility of the modified offline strategy as $u'$. Then, we will show that
\begin{equation}\label{step1}
\frac{u^*}{u'} \leq 1+\frac{1}{\alpha R-\beta-1}.
\end{equation}
In the second step, we will \textcolor{black}{prove} that
\begin{equation}\label{step2}
u' \leq u^o.
\end{equation}
Then, inequality (\ref{step1}) and inequality (\ref{step2}) complete the proof.

\textit{Step 1:}
Different from the offline optimal strategy, at the beginning of each cache period, the modified offline strategy releases all contents and immediately purchases and caches the contents that the offline optimal strategy decides to cache in the cache period. The releasing introduces additional purchasing and transmission costs. Consider a content, say content $n$, that is cached by the offline optimal strategy in cache period $l$. Then content $n$ will also be cached by the modified offline strategy. Let $u^*_n(l)$ and $u'_n(l)$ denote the utility of content $n$ in cache period $l$ by executing the offline optimal strategy and the modified offline strategy, respectively. First, we have
\begin{equation}
\begin{split}
u_n^*(l)=&R'_n(l)[p_n^s(l)+s_n C_d]-v_n(l)s_n C_a\\
&-x_n(l)(p_n+s_n C_d) \\
\leq &R'_n(l)[p_n^s(l)+s_n C_d]-v_n(l)s_n C_a
\nonumber
\end{split}
\end{equation}
and
\begin{equation}
u_n'(l)=R'_n(l)[p_n^s(l)+s_n C_d]-v_n(l)s_n C_a-(p_n+s_n C_d).
\nonumber
\end{equation}
In the above expression, $p_n+s_n C_d$ in $u_n'(l)$ means the purchasing and transmission costs, explained as follows. If the content $n$ has not been cached at the end of the previous cache period by the offline optimal strategy, the modified offline strategy should purchase and cache it. If the content $n$ has been cached at the end of the previous cache period by the offline optimal strategy, the modified offline strategy should release it and immediately repurchase and cache it.

From the above expressions of $u_n^*(l)$ and $u_n'(l)$, we have
\begin{equation}
\begin{split}
\frac{u^*_n(l)}{u_n'(l)} &\leq \frac{R'_n(l)[p_n^s(l)+s_n C_d]-v_n(l)s_n C_a}{R'_n(l)[p_n^s(l)+s_n C_d]-v_n(l)s_n C_a -p_n-s_n C_d}\\
&=1+\frac{p_n+s_n C_d}{R'_n(l)[p_n^s(l)+s_n C_d]-v_n(l)s_n C_a-p_n-s_n C_d}.
\nonumber
\end{split}
\end{equation}
Together with the fact $v_n(l)\leq b$ (i.e., the caching duration of content $n$ in a cache period is not more than the total length of the cache period), we have
\begin{equation}
\begin{split}
\frac{u^*_n(l)}{u_n'(l)} &\leq 1+\frac{p_n+s_n C_d} {R'_n(l)[p_n^s(l)+s_n C_d]-b s_n C_a-p_n-s_n C_d}.
\nonumber
\end{split}
\end{equation}
Let $R=\min\limits_{n\in\mathcal{N}_c,\textcolor{black}{t\in \mathcal{T}}}R'_n(l)$ denote the minimum number of requests for a cached content in a single cache period. So,
\begin{equation}\label{ration_of_u*}
\begin{split}
\frac{u^*_n(l)}{u_n'(l)} &\leq 1+\frac{p_n+s_n C_d} {R[p_n^s(l)+s_n C_d]-b s_n C_a-p_n-s_n C_d}.
\end{split}
\end{equation}
Besides, denote $\frac{p_n^s(l)+s_n C_d}{p_n+s_n C_d}$ as the profit ratio of content $n$, which is the ratio of the profit in response to one request of content $n$ to the purchasing and transmission costs of content $n$ (from CPs), and denote the minimum profit ratio among all cached contents as $\alpha=\min\limits_{n\in\mathcal{N}_c,
\textcolor{black}{t\in \mathcal{T}}} \frac{p_n^s(l)+s_n C_d}{p_n+s_n C_d}$. Then, denote $\frac{b s_n C_a}{p_n+s_n C_d}$ as the cost ratio of content $n$, which is the ratio of the caching cost for one cache period of content $n$ to the purchasing and transmission costs of content $n$, and denote the maximum cost ratio among all cached contents as $\beta=\max\limits_{n\in\mathcal{N}_c, \textcolor{black}{t\in \mathcal{T}}}\frac{b s_n C_a}{p_n+s_n C_d}$.
Thus, from (\ref{ration_of_u*}) we have
\begin{equation}
\begin{split}
\frac{u^*_n(l)}{u_n'(l)}\leq 1+\frac{1}{\alpha R-\beta -1}.
\nonumber
\end{split}
\end{equation}
Then for all contents across all cache periods, the utilities of the above offline optimal strategy and modified offline strategy satisfy the following inequality:
\begin{equation}
\begin{split}
u^*=\sum_{l \in \mathcal{L}}\sum_{n \in \mathcal{N}}u^*_n(l)&\leq \sum_{l \in \mathcal{L}}\sum_{n \in \mathcal{N}}(1+\frac{1}{\alpha R-\beta -1})u'_n(l)\\
&=(1+\frac{1}{\alpha R-\beta -1})u',
\nonumber
\end{split}
\end{equation}
which leads to inequality (\ref{step1}).

\textit{Step 2:}
In this step, we introduce a {\it greedy offline strategy} with utility in cache period $l$ denoted by $u''(l)$.
At the beginning of cache period $l$, the greedy offline strategy releases all contents, and performs purchasing and caching in cache period $l$ so as to maximize  the ENSP's utility in cache period $l$.\footnote{Such purchasing and caching decisions in the greedy offline strategy can be obtained by solving problem $\mathcal{P}2$ for cache period $l$, given that no content is cached by the ENSP at the beginning of cache period $l$.}
Note that in both greedy offline strategy and modified offline strategy, the ENSP has no cached content at the beginning of each cache period. The greedy offline strategy focuses on utility maximization in the current cache period, while the modified offline strategy follows the offline optimal strategy that focuses on long-term utility maximization over all cache periods. So, for a specific cache period, say cache period $l$, the greedy offline strategy's utility is not less than that of the modified offline strategy denoted by $u'(l)$, i.e., $u'(l)\leq u''(l)$.

Recall that the DT-OCA-PP always selects the contents yielding the maximum utility in cache period $l$ denoted by $u^o(l)$. Thus, its utility in a cache period is not less than that of the greedy offline strategy in the cache period, since the DT-OCA-PP does not need to pay the purchasing and transmission costs for contents that have been cached by the ENSP at the end of the previous cache period.
So, we have $u''(l)\leq u^o(l)$, and then $u'(l) \leq u^o(l)$.
Therefore, $u'=\sum_{l\in \mathcal{L}}u'(l) \leq \sum_{l\in \mathcal{L}}u^o(l)=u^o$, which is exactly the inequality (\ref{step2}).
\end{proof}
\begin{remark}
\textcolor{black}{
The ENSP expects to purchase and cache popular contents only. For popular contents, the value of $R$ (i.e., the minimum number of requests for a content in a cache period) is expected to be a large number. Thus, the upper bound of competitive ratio given in (\ref{eq:theorem}) is expected to be close to $1$. For example, for the simulation scenarios in Section \ref{Sec:simulation}, the upper bound is calculated as $1.34$, which is quite promising.
}
\end{remark}

\section{Experimental Results}\label{Sec:simulation}
In this section, we evaluate the performance of the proposed DT-OCA via extensive simulations.
\subsection{Parameter Setting}
We simulate an edge caching network with the following default setting (unless otherwise specified). $T=30,000$, $L=3,000$, $b=10$, $\tau=\SI{1}{\second}$. In CPs, totally $300,000$ contents are generated uniformly within the $30,000$ time slots.
The freshness threshold $\varphi$ is set to be $30$ time slots.
The content prices sold by CPs and the content sizes are uniformly valued in $[20, 200]$ and $[2, 50]$ units, respectively.
$S_{\max}$ is set as $300$ units, and the parameters associated with the utility model are set as $p_n^{\max}=30$, $\lambda=1$, $C_d=1$, and $C_a=0.1$.
The updating frequency of the DT network is $\SI{1}{Hz}$.
The Transformer in the DT-based prediction method is implemented by PyTorch with parameters $H=4$ and $h=12$.
The experimental results are obtained by averaging over all cache periods, except for the results shown in Fig.~\ref{Fig:cached_content}, Fig.~\ref{Fig:dataset_artificial}, and Fig.~\ref{Fig:dataset_meme}.

We also simulate the following benchmark algorithms for comparison.
\begin{itemize}
\item \textit{Optimal-selection Predictive Least Frequently Used (OP-LFU)} \cite{OPLFU}: LFU maintains an ordered list to track the number of content requests and caches as many contents as possible. When there is no sufficient cache space, the ENSP releases the least frequently used contents. \textcolor{black}{Compared to LFU, OP-LFU predicts future request information and determines caching strategy based on the predicted results while LFU utilizes the past information. In the experiments, we adopt the DT-based prediction method and LSTM model in OP-LFU, named OP-LFU (TRM) and OP-LFU (LSTM), respectively.}
\item \textit{Window-Least Frequently Used (W-LFU)} \cite{WLFU}: Different from LFU, W-LFU only records content requests for a time window, and determines its caching strategy based on the recorded content request history in the time window.
\item \textit{First In First Out (FIFO)} \cite{FIFO}: This algorithm releases the earliest generated contents and caches new contents as many as possible.
\item \textcolor{black}{\textit{Following the Perturbed Leader (FTPL)} \cite{FTPL}: For each content, FTPL assigns a random initial utility when the content is generated. At each cache period, FTPL evaluates the total accumulated utility of each content in $\mathcal{N}_p\cup\mathcal{N}_c$ during the past cache periods, and selects the contents with highest accumulated utility to cache.}
\item \textit{Random}: This algorithm randomly selects contents to cache from the purchasable content set.
\end{itemize}

To evaluate the performance of our algorithm and the benchmark algorithms, we use three metrics: the utility of the ENSP, \textit{hit rate}, and \textit{average AoI}.
The hit rate is defined as the ratio of the number of requests for contents cached by the ENSP to the total number of requests.
The average AoI is calculated according to (\ref{average AoI}).


\subsection{Performance Evaluation and Comparison}\label{Sec:feasibility}
\begin{figure}[t]
  \centering
  \includegraphics[width=0.32\textwidth]{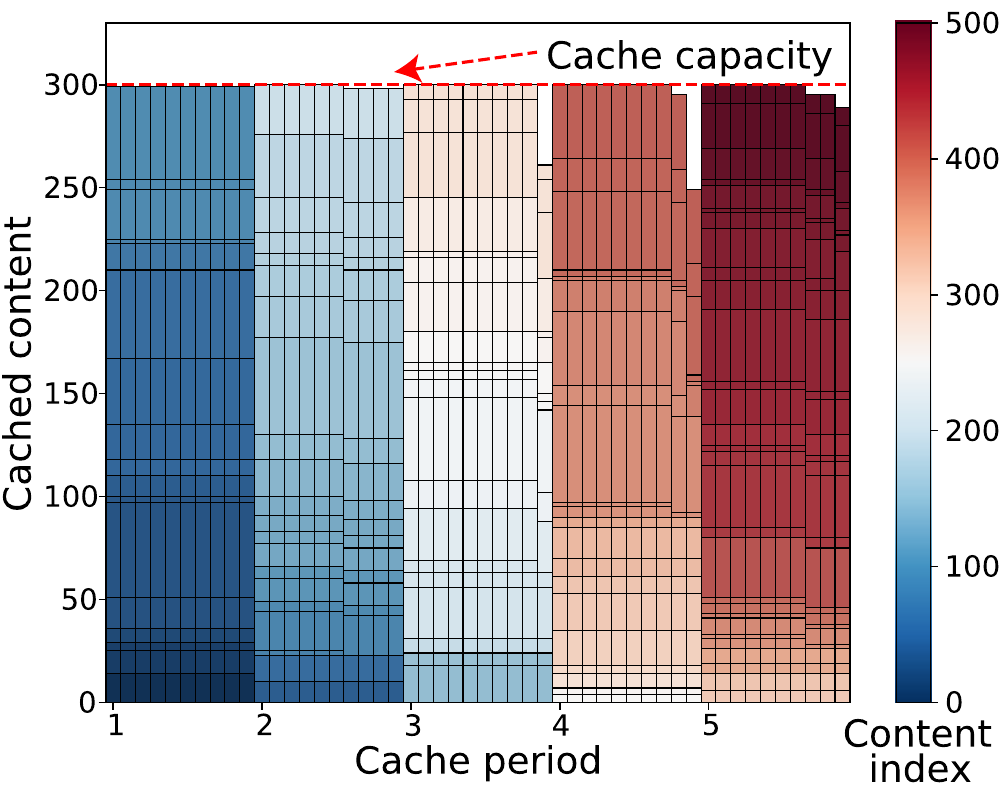}
  \caption{ENSP's cache status by using DT-OCA in the first $5$ cache periods.} \label{Fig:cached_content}
\end{figure}

\begin{figure*}[t]
    \centering
        {\includegraphics[height=3.5cm]{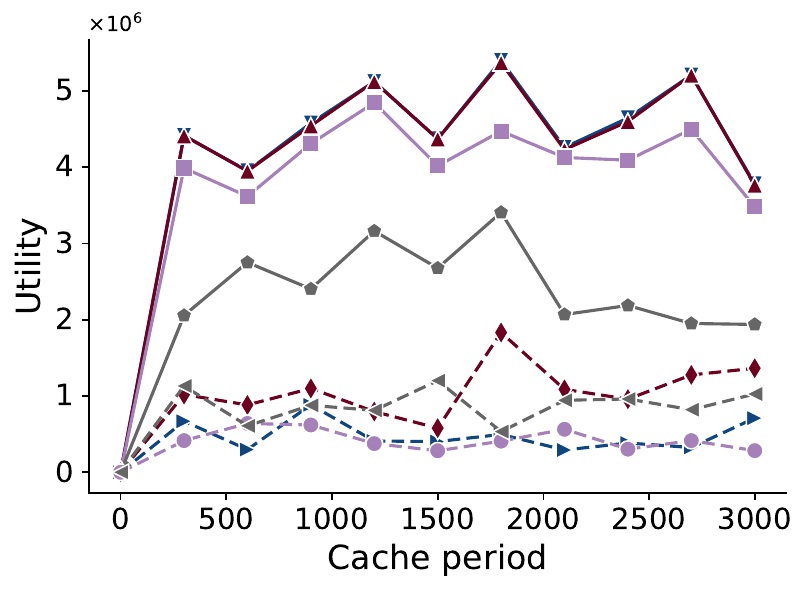}}
        {\includegraphics[height=3.5cm]{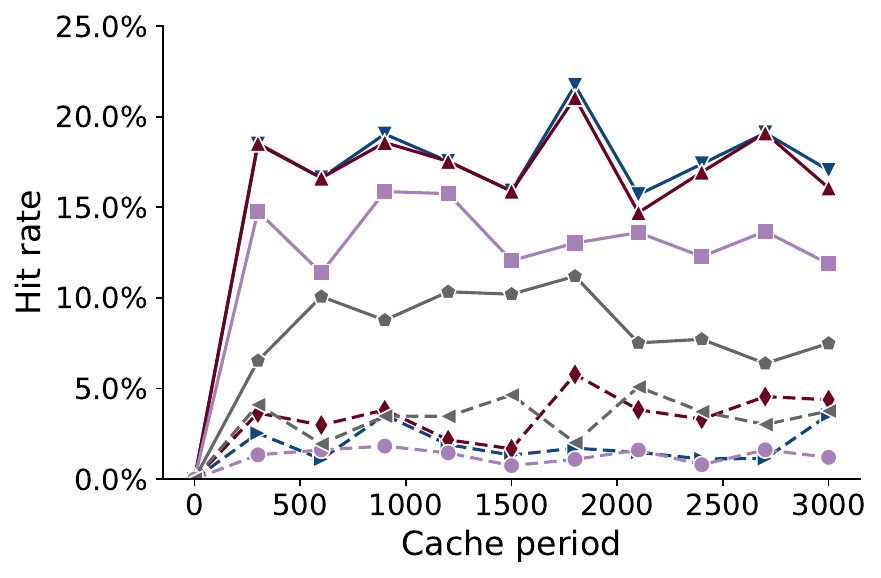}}
        {\includegraphics[height=3.5cm]{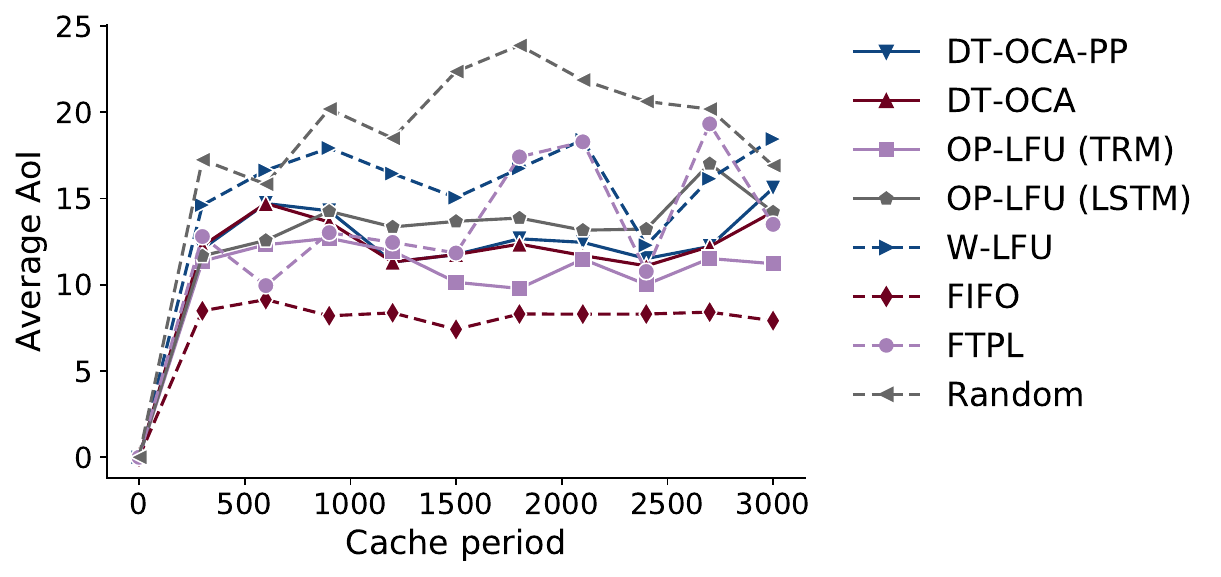}}
        \textcolor{black}{\caption{(a) The utility, (b) the hit rate, and (c) the average AoI.}\label{Fig:dataset_artificial}}
\end{figure*}


By using DT-OCA, Fig.~\ref{Fig:cached_content} shows the cache status of the ENSP over the first $5$ cache periods. Roughly 500 contents are generated in the first 5 cache periods, and thus, the content indices in Fig.~\ref{Fig:cached_content} are $0\sim 500$. It is seen that the cache space is almost fully utilized in each cache period. Fig.~\ref{Fig:cached_content} also shows that the occupied cache space does not become larger within a cache period. This is because content purchasing only happens at the beginning of a cache period. Over the totally $3,000$ cache periods, DT-OCA and DT-OCA-PP occupy $99.98\%$ of the cache space, while other benchmark algorithms occupy $99.88\%$.

Fig.~\ref{Fig:dataset_artificial} depicts the utility, hit rate, and average AoI of different algorithms in the $3,000$ cache periods.

Fig.~\ref{Fig:dataset_artificial}\;(a) shows that, in terms of utility, the DT-OCA-PP and DT-OCA have almost the same performance. This means that the prediction given by the Transformer in the DT-based prediction method could achieve very similar results to the case with perfect prediction. Further, the utility of the DT-OCA is much higher than those of other benchmark algorithms. This is because DT-OCA adaptively develops the caching strategy using predicted future content popularity and takes into account content freshness and content size.
OP-LFU (TRM) that uses our DT-based prediction method provides higher utility than OP-LFU (LSTM), which also illustrates the effectiveness of using the Transformer in the DT-based prediction method. FIFO only focuses on the freshness of contents while ignoring the popularity of contents. Random relies completely on random choices and is not intelligent. W-LFU and FTPL utilize the previous request and previous utility information that cannot adapt to future popularity changes. 

Fig.~\ref{Fig:dataset_artificial}\;(b) shows that DT-OCA-PP and DT-OCA have the highest hit rate, respectively. The hit rate of W-LFU, FIFO, FTPL, and Random are less than our DT-OCA, because they do not predict future content popularity when determining their caching strategies.
OP-LFU (TRM) and OP-LFU (LSTM) adopt content popularity prediction, but do not take into account the content size, and thus, cannot perfectly utilize the cache space. Specifically, they may cache popular contents with a huge size, and therefore, cache much fewer contents than our DT-OCA.

Fig.~\ref{Fig:dataset_artificial}\;(c) shows that the average AoI of DT-OCA and DT-OCA-PP is less than those of OP-LFU (LSTM), W-LFU, FTPL, and Random. However, DT-OCA and DT-OCA-PP have slightly higher average AoI than OP-LFU (TRM) and FIFO, due to the following reasons. OP-LFU (TRM) caches the most popular contents. Since fresh contents normally attract high interests from users, most popular contents are very likely to be those with higher freshness (i.e., lower AoI). Thus, OP-LFU (TRM) achieves low average AoI. FIFO caches only freshest contents, and thus, has the lowest average AoI among all the algorithms. Nevertheless, the utility and hit rate of DT-OCA are higher than those of OP-LFU (TRM) and FIFO, as shown in Fig.~\ref{Fig:dataset_artificial}\;(a) and Fig.~\ref{Fig:dataset_artificial}\;(b).
\textcolor{black}{In subsequent experiments, we do not implement OP-LFU (LSTM), since OP-LFU (TRM) outperforms OP-LFU (LSTM) in utility, hit rate, and average AoI as shown in Fig.~\ref{Fig:dataset_artificial}\;(a), Fig.~\ref{Fig:dataset_artificial}\;(b), and Fig.~\ref{Fig:dataset_artificial}\;(c), respectively.}

\begin{figure*}[t]
    \centering
        {\includegraphics[height=3.5cm]{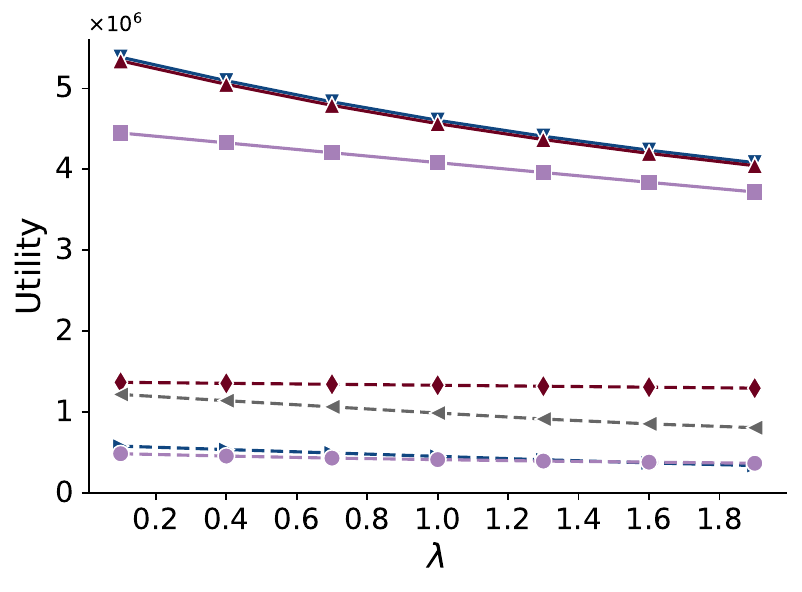}}
        {\includegraphics[height=3.5cm]{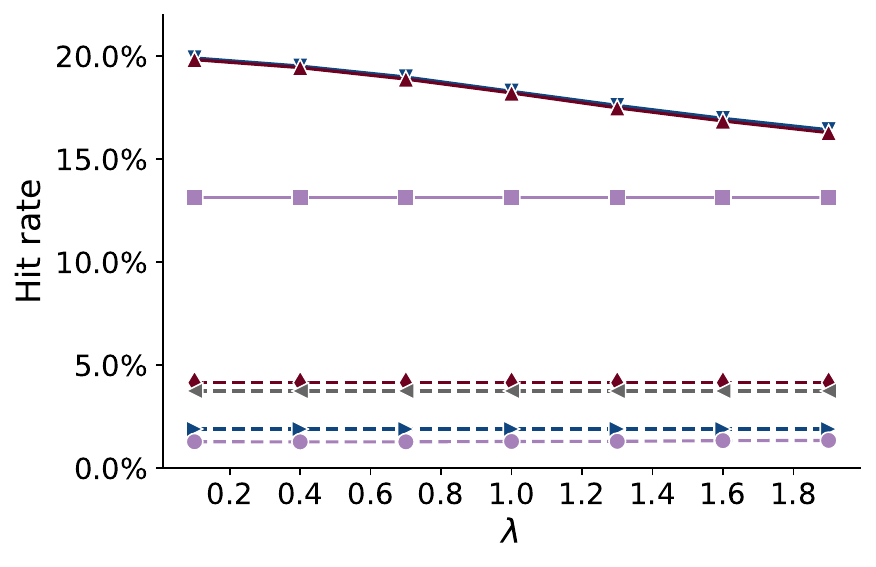}}
        {\includegraphics[height=3.5cm]{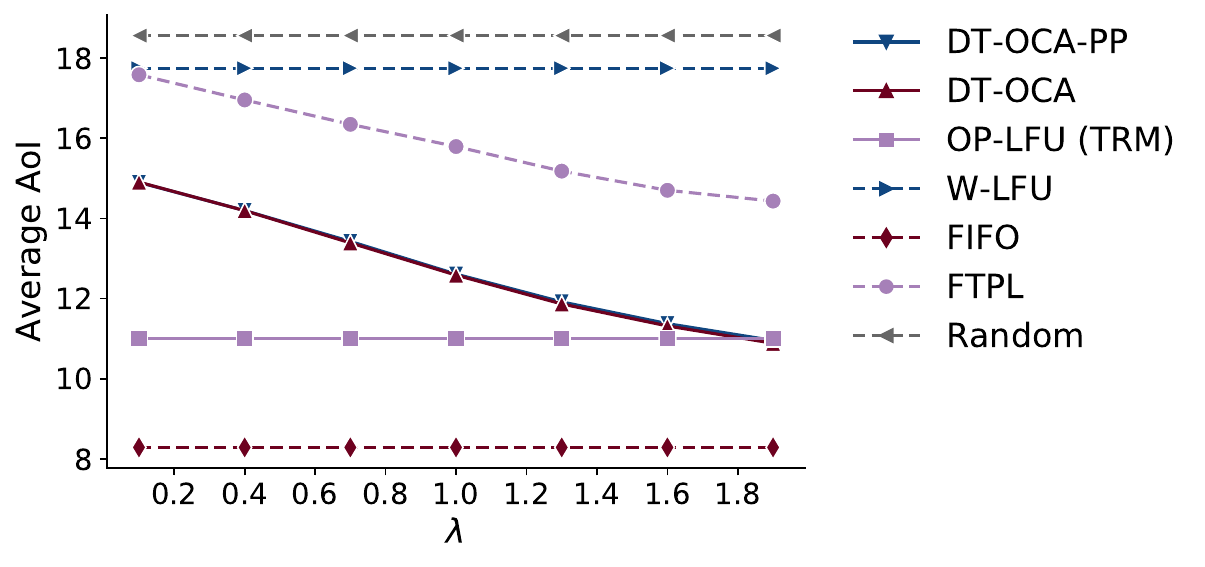}}
        \textcolor{black}{\caption{(a) The impact of $\lambda$ on utility, (b) the impact of $\lambda$ on hit rate, and (c) the impact of $\lambda$ on average AoI.}\label{Fig:lambda}}
    \vspace{20pt}
    \centering
        {\includegraphics[height=3.5cm]{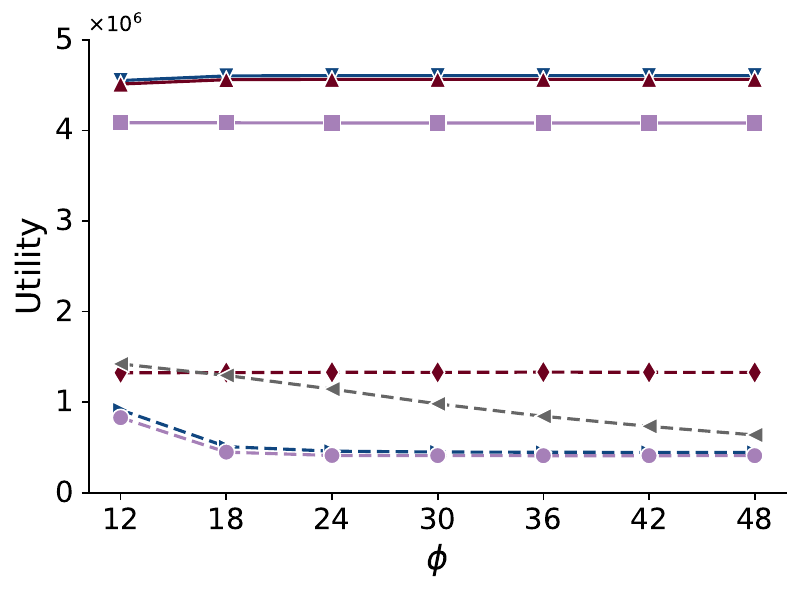}}
        {\includegraphics[height=3.5cm]{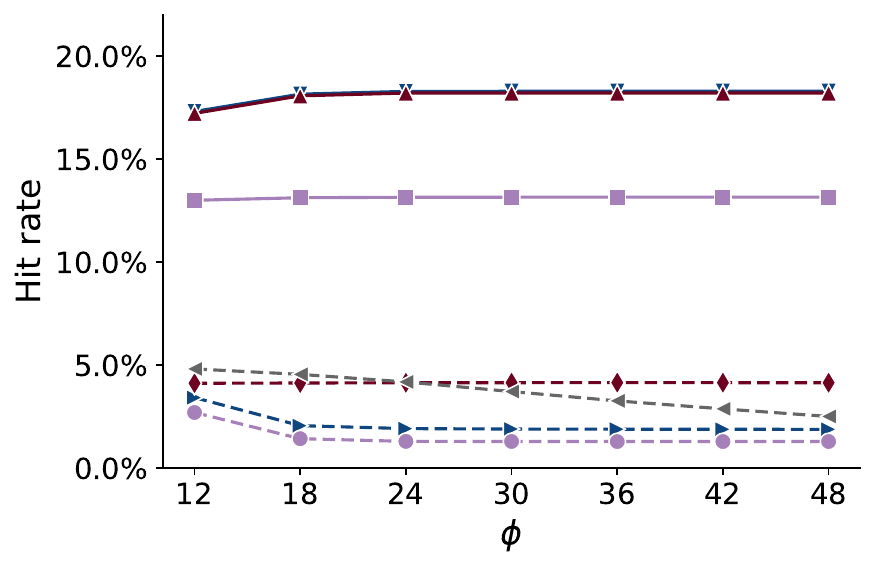}}
        {\includegraphics[height=3.5cm]{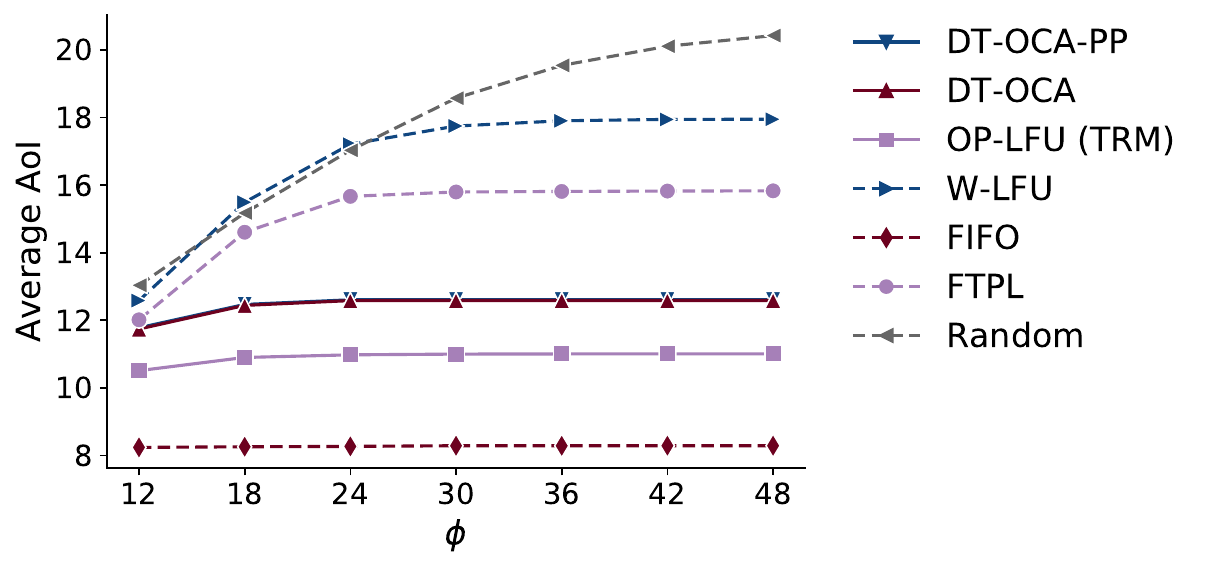}}
        \textcolor{black}{\caption{(a) The impact of $\phi$ on utility, (b) the impact of $\phi$ on hit rate, and (c) the impact of $\phi$ on average AoI.}\label{Fig:phi}}
    \vspace{20pt}
    \centering
        {\includegraphics[height=3.5cm]{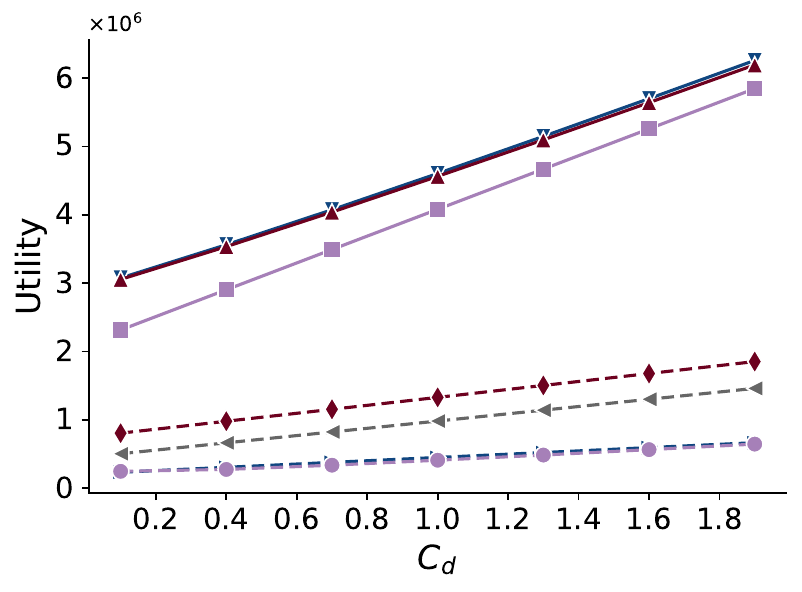}}
        {\includegraphics[height=3.5cm]{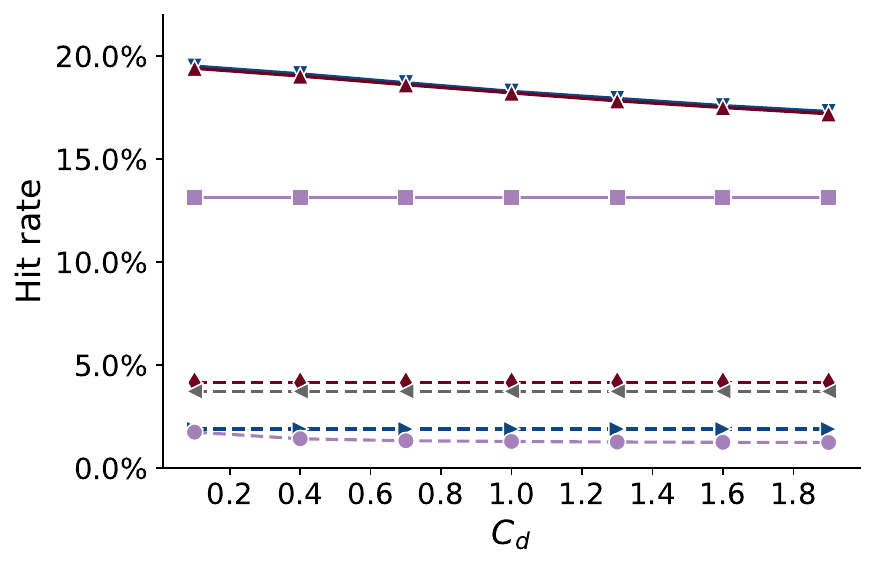}}
        {\includegraphics[height=3.5cm]{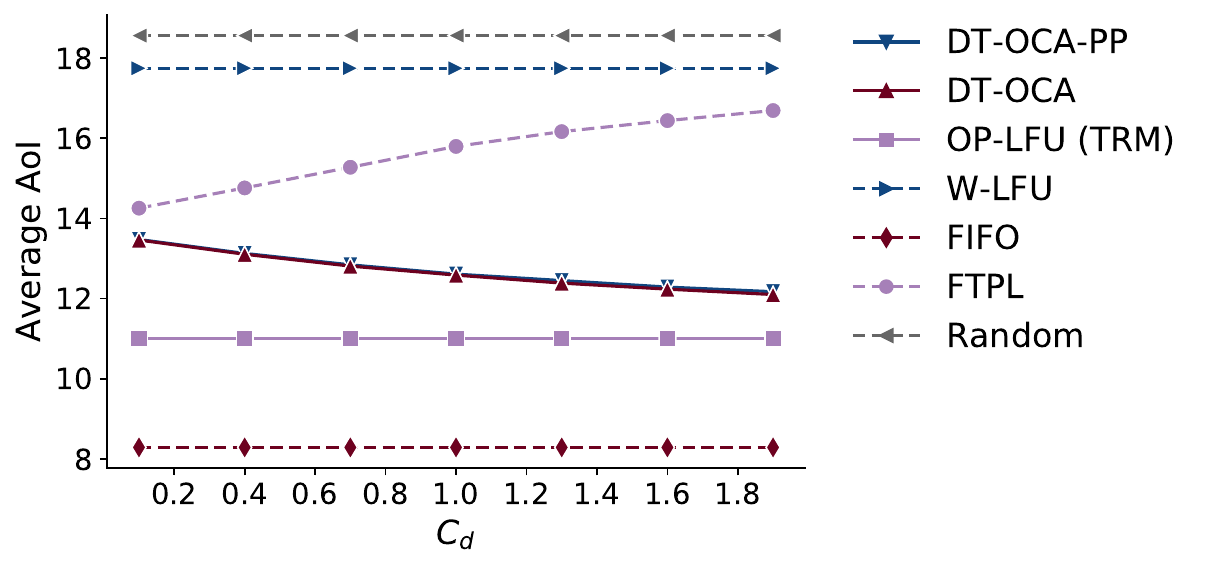}}
        \textcolor{black}{\caption{(a) The impact of $C_d$ on utility, (b) the impact of $C_d$ on hit rate, and (c) the impact of $C_d$ on average AoI.}\label{Fig:C_d}}
\end{figure*}

Next, we evaluate the performance of DT-OCA and other benchmark algorithms under different system parameters.

Fig.~\ref{Fig:lambda} shows the utility as well as the hit rate and average AoI when weight parameter $\lambda$ (weight of average AoI when determining the service fee, as shown in (\ref{service fee known})) changes. For all algorithms, when $\lambda$ becomes larger, the service fee charged by the ENSP to users decreases, and therefore,  the utility of each content decreases. Thus, the utility of all algorithms decrease.
For DT-OCA and DT-OCA-PP, when $\lambda$ becomes larger, the hit rate decreases, and the average AoI also decreases, due to the following reason. With a large $\lambda$, the service fee decreases a lot when AoI increases slightly. So the ENSP prefers to cache most recently generated contents. However, if we consider a short time window starting from content generation moment, those most recently generated contents actually do not attract a lot of requests from users, as it takes time for users to react to the new contents. Thus, when $\lambda$ becomes larger, DT-OCA and DT-OCA-PP have lower hit rate and lower average AoI.
When $\lambda$ becomes larger, the hit rate and average AoI of OP-LFU (TRM), W-LFU, FIFO, and Random keep the same, as their caching strategies are not based on service fee.


Fig.~\ref{Fig:phi} demonstrates the impact of threshold parameter $\phi$ that determines the range of purchasable contents. For DT-OCA-PP and DT-OCA, when $\phi$ increases from $12$ to $18$, more contents are eligible for purchasing and caching, and thus, the utility of the ENSP slightly increases, and the hit rate also slightly increases. The average AoI increases, since contents with high AoI may be purchased and cached. When $\phi$ keeps increasing, the utility, hit rate, and average AoI almost keep the same. This is because the two algorithms still prefer to purchase and cache fresher contents. Varying of $\phi$ does not affect performance of FIFO, as FIFO only caches the freshest contents. For Random, W-LFU, and FTPL, when $\phi$ increases, their utility and hit rate decrease, and average AoI increases, because less fresh contents are eligible for purchasing and caching. For OP-LFU (TRM), varying of $\phi$ has no big impact on utility, hit rate, and average AoI, for the following reason. Although larger $\phi$ means that less fresh contents may be cached, OP-LFU (TRM) intends to release least frequently used contents (when cache space is insufficient), which are very likely with less freshness.


Fig.~\ref{Fig:C_d} shows the impact of $C_d$ (the transmission cost per unit size from CPs to the ENSP). The unit caching cost $C_a$ is set to one-tenth of $C_d$. When $C_d$ becomes larger, the utility of all algorithms increase, because larger $C_d$ means more savings for transmission cost from CPs to the ENSP. For DT-OCA and DT-OCA-PP, when $C_d$ becomes larger, the ENSP prefers to cache larger-size contents, and thus, the number of contents decreases, leading to lower hit rate as shown in Fig.~\ref{Fig:C_d}\;(b). On the other hand, since large-size contents need to be cached, the ENSP needs to release small-size contents with less freshness, and thus, the average AoI decreases as shown in Fig.~\ref{Fig:C_d}\;(c). Varying of $C_d$ does not change hit rate and average AoI of OP-LFU (TRM), W-LFU, FIFO, and Random, as these algorithms do not take into account transmission cost from CPs to the ENSP when determining their caching strategies.


In DT-OCA and DT-OCA-PP, the DT network should be periodically updated.
Fig.~\ref{Fig:dis_sam} depicts the effect of updating frequency of the DT network on the utility of the ENSP. Here updating frequency means the number of updates per time slot. Generally, when updating frequency becomes smaller, the accuracy of the DT network becomes less, and thus, the utility of the ENSP becomes smaller. However, interestingly and counter-intuitively, when the updating frequency changes from $1/4$ to $1/5$, the utility of the ENSP increases, explained as follows. When the updating frequency is $1/4$, it means that there is an update every $4$ time slots. Recall that each cache period has $10$ time slots and the purchasing decision is made at the beginning of a cache period. Thus, a purchasing decision may be made based on updating a few time slots ago. This {\it a-few-time-slot time gap} may degrade the system performance. On the other hand, when the updating frequency is $1/5$, there is always an update immediately before each purchasing decision. In other words, the {\it time gap} is zero, and thus, there is no system degradation as the case with updating frequency being $1/4$.


\begin{figure}[t]
    \centering
    {\includegraphics[width=0.32\textwidth]{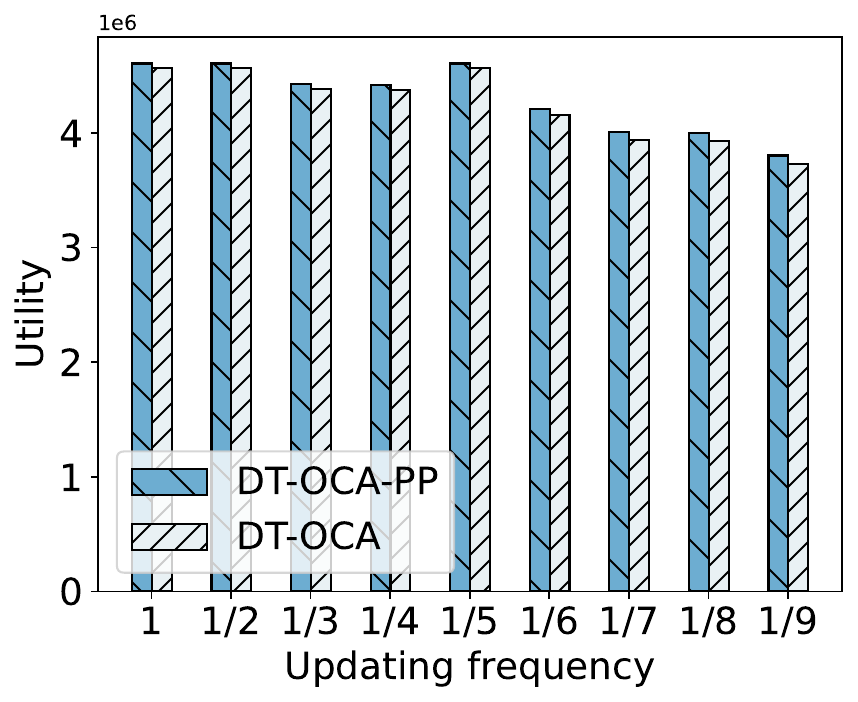}}
    \caption{The impact of updating frequency of the DT network on utility.}\label{Fig:dis_sam}
\end{figure}

\vspace{-5mm}
\textcolor{black}{\subsection{Performance with Real-World Dataset}}
\label{sec:Meme}
\begin{figure*}[t]
    \centering
        {\includegraphics[height=3.5cm]{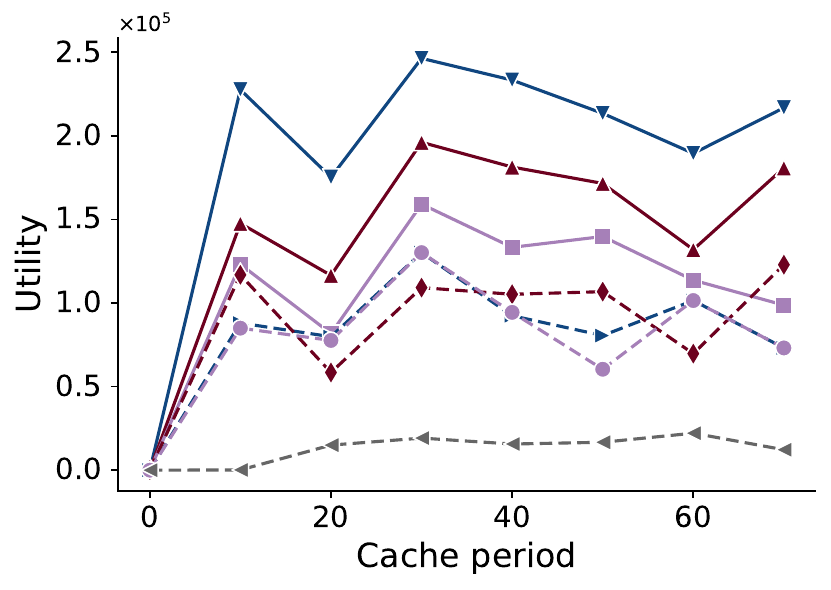}}
        {\includegraphics[height=3.5cm]{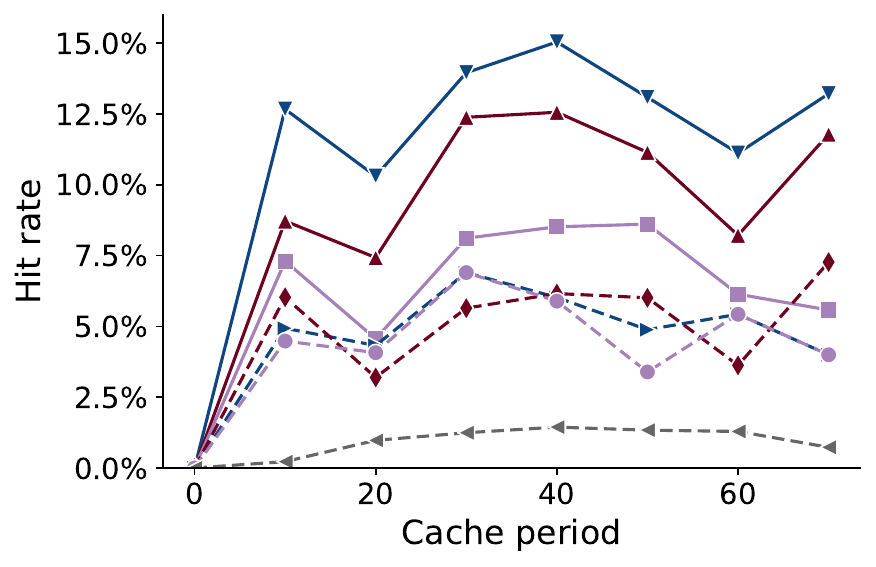}}
        {\includegraphics[height=3.5cm]{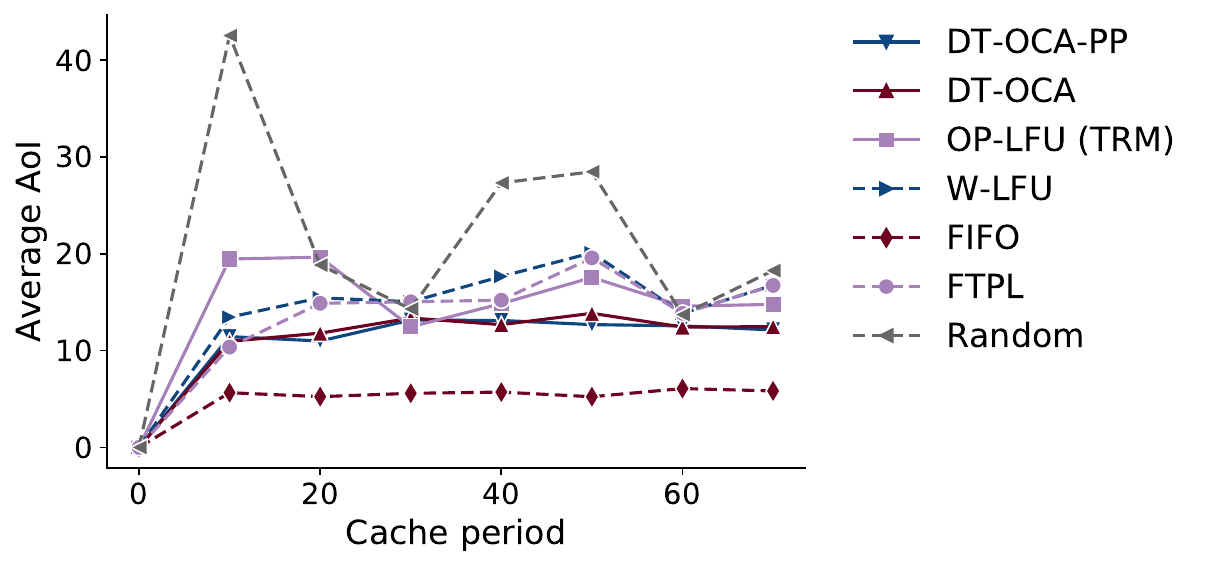}}
        \textcolor{black}{\caption{(a) The utility on MemeTracker dataset, (b) the hit rate on MemeTracker dataset, and (c) the average AoI on MemeTracker dataset.}\label{Fig:dataset_meme}}
\end{figure*}
We also evaluate the performance of our algorithm and the benchmark algorithms by using a real-world dataset, MemeTracker dataset: ``\textit{MemeTracker builds maps of the daily news cycle by analyzing around 900,000 news stories and blog posts per day from 1 million online sources, ranging from mass media to personal blogs}" \cite{memepaper}.

From this dataset, we filter out the majority of contents that almost no one asked for and select $7,476$ contents for our experiments. We set a two-timescale of $T=750$, $L=75$, $b=10$, and $\tau=\SI{1}{\hour}$. The freshness threshold $\varphi$ is set as $50$ time slots and the utility model-related parameters are set as $p_n^{\max}=100$ and $\lambda=1.5$. The Transformer parameters in the DT-based prediction method are set to $H=4$ and $h=8$.

The performance of the algorithms on the MemeTracker dataset are shown in Fig.~\ref{Fig:dataset_meme}. From this figure, similar observations can be made to the observations from Fig.~\ref{Fig:dataset_artificial}, i.e., compared with other benchmark algorithms, DT-OCA and DT-OCA-PP have larger utility and higher hit rate. They have lower average AoI than other algorithms except FIFO (but FIFO has much smaller utility and hit rate than DT-OCA and DT-OCA-PP). The performance gap between DT-OCA and DT-OCA-PP in Fig.~\ref{Fig:dataset_meme} is a bit larger than the gap in Fig.~\ref{Fig:dataset_artificial}, due the the following reason. Our DT-OCA is designed for caching request-intensive contents. However, most of the contents in the MemeTracker dataset are not popular enough for a request-intensive scenario, and thus, the dataset may be insufficient for the prediction model training, which leads to a slightly larger gap between DT-OCA and DT-OCA-PP. Unfortunately, we were not able to find more appropriate real-world dataset in the open public. We expect a smaller gap between DT-OCA and DT-OCA-PP when more appropriate real-world datasets become available in the future.

\vspace{3mm}
\section{Conclusion}\label{Sec:conclusion}
In this paper, we investigate the ENSP, which purchases and caches contents from CPs and resells them to users. To maximize the utility of the ENSP by taking into account content size, content popularity, and content freshness, the formulated problem is non-convex and NP-hard. We propose DT-OCA to solve the problem. DT-OCA decomposes the formulated problem into a series of subproblems, each for a cache period. For each subproblem, DT-OCA first uses a DT-based prediction method to predict future content popularity, and based on the prediction, it develops the caching strategy for the cache period. The competitive ratio of DT-OCA is analyzed. Extensive experimental results demonstrate the effectiveness and superiority of DT-OCA over other benchmark algorithms.
In future work, we may further investigate collaboration among ESs, competition among CPs, budgets of users, etc.

{
\footnotesize
\bibliographystyle{IEEEtranN}
\bibliography{mybib}
}

\end{document}